\documentclass[a4paper,twoside]{article}

\usepackage{graphicx,xcolor,textpos}

\graphicspath{{./img/}{./Images/}}

\usepackage{helvet}
\usepackage[english]{babel}
\usepackage[normalem]{ulem}
\usepackage{amsmath}
\usepackage{amsthm}
\usepackage{thmtools,thm-restate}
\usepackage{bbm}
\usepackage{amssymb}
\usepackage{hyperref}
\usepackage{relsize}
\usepackage[margin=0.9in]{geometry}
\usepackage{physics}
\usepackage[shortlabels]{enumitem}
\usepackage{mathtools}
\usepackage{changepage}
\usepackage{caption}
\usepackage{subcaption}
\usepackage{verbatim}
\usepackage{url}
\usepackage{standalone}
\usepackage{tikz}
\usetikzlibrary{calc,patterns,angles,quotes}
\usepackage{dsfont}
\usepackage{mathrsfs}
\usepackage[affil-it]{authblk}
\usepackage{titlesec}

\makeatletter
\newcommand\Label[1]{&\refstepcounter{equation}(\theequation)\ltx@label{#1}&}
\makeatother

\newcommand{\stkout}[1]{\ifmmode\text{\sout{\ensuremath{#1}}}\else\sout{#1}\fi}

\newcommand{\fk}{\mathfrak}

\newcommand{\sr}[1]{\text{sr}(#1)}

\let\originalleft\left
\let\originalright\right
\renewcommand{\left}{\mathopen{}\mathclose\bgroup\originalleft}
\renewcommand{\right}{\aftergroup\egroup\originalright}

\newcommand{\UU}{\mathcal U}

\newcommand{\HH}{\mathcal H}
\newcommand{\ZZ}{\mathbb Z}

\renewcommand{\AA}{\mathcal A}

\newcommand{\RR}{\mathbb R}
\newcommand{\NN}{\mathbb{N}}
\newcommand{\id}{\mathbbm{1}}

\newcommand{\caH}{\mathcal{H}}
\newcommand{\caA}{\mathcal{A}}
\newcommand{\caB}{\mathcal{B}}

\newcommand{\bbZ}{\mathbb{Z}}

\newcommand{\bbN}{\mathbb{N}}

\renewcommand{\AA}{\mathcal A}

\renewcommand{\RR}{\mathbb R}

\newcommand{\dist}{\text{dist}}
\newcommand{\dm}[1]{\mathscr P({#1})}
\newcommand{\kb}[1]{\ket{#1}\bra{#1}}

\DeclareMathOperator{\Z}{\mathbb{Z}}

\newcommand{\Ad}[1]{\textrm{Ad}\left(#1\right)}

\newtheoremstyle{exampstyle}
{15pt} 
{15pt} 
{\itshape} 
{} 
{\bfseries} 
{.} 
{.5em} 
{} 

\theoremstyle{exampstyle}

\newtheorem{theorem}{Theorem}[section]

\newtheorem{lemma}[theorem]{Lemma}

\newtheorem{corollary}[theorem]{Corollary}

\newtheorem{proposition}[theorem]{Proposition}

\def\BigglVert{\mkern-.23mu\Biggl|\mkern-2.3mu\Biggl|\mkern-.25mu}
\def\BiggrVert{\mkern-.23mu\Biggr|\mkern-2.3mu\Biggr|\mkern-.25mu}

\newcommand{\cutoff}{\mathfrak J}

\title{Pure gapped ground states of spin chains are short-range entangled}

\author[2]{Wojciech De Roeck \thanks{email: \href{wojciech.deroeck@kuleuven.be} {wojciech.deroeck@kuleuven.be}}}
\author[1]{Martin Fraas \thanks{email: \href{mfraas@ucdavis.edu}{mfraas@ucdavis.edu}}}
\author[2]{Bruno de O. Carvalho \thanks{email: \href{bruno.oliveira@kuleuven.be}{bruno.oliveira@kuleuven.be}}}

\affil[1]{Department of Mathematics, University of California, Davis, Davis, CA, 95616, USA}
\affil[2]{Instituut voor Theoretische Fysica, KU Leuven, 3001 Leuven, Belgium}

\date{\today}

\begin{document}

\maketitle

\begin{abstract}
We consider spin chains with a finite range Hamiltonian. For reasons of simplicity, the chain is taken to be infinitely long. A ground state is said to be a unique gapped ground state if its GNS Hamiltonian has a unique ground state, separated by a gap from the rest of the spectrum.
By combining some powerful techniques developed in the last years, we prove that each unique gapped ground state is \emph{short-range entangled}: It can be mapped into a product state by a finite time evolution map generated by a Hamiltonian with exponentially quasi-local interaction terms.  This claim makes precise the common belief that one-dimensional gapped systems are topologically trivial in the bulk.
\end{abstract}

\section{Introduction}

In this work we show that a ground state $\psi$ of a gapped one-dimensional spin chain is short-range entangled (SRE). More precisely, we construct a quasi-local automorphism $\alpha$, generated by an interaction with exponentially decaying tails, such that
\begin{equation}
\psi = \phi \circ \alpha,
\end{equation}
for a product state $\phi$. This establishes, in rigorous form, the widely held expectation that in one dimension there are no intrinsically topological phases of matter in the absence of symmetry protection.~Despite the ubiquity of this statement in the physics literature \cite{chen_gu_wen_2011,schuch2011MatrixProduct}, a complete proof has until now remained unavailable. We begin by outlining a brief history of the problem and by identifying the principal obstruction that prevented earlier arguments from giving a full proof.

The question has also been of central interest in the quantum information community, largely through attempts to justify the empirical success of the density matrix renormalization group (DMRG) algorithm. From the modern perspective, this is equivalent to proving that the ground state admits a matrix product state (MPS) representation providing accurate local approximations. The first such construction was given in \cite{Verstraete-Cirac} and subsequently refined in e.g. \cite{Schuch-Verstraete} and \cite{Dalzell-Brandao}. These works produce an MPS with bond dimension $D$ whose restriction to intervals of length $l$ is $\epsilon$-close to the corresponding restriction of the ground state. The bounds differ in how the bond dimension scales with $\epsilon$ and $l$; the currently sharpest result is contained in \cite[Theorem~2]{Dalzell-Brandao}.

The statements above are, however, local in nature and do not imply a global approximation of the ground state by an MPS. The obstruction to a global result was already recognized in these works. Paraphrasing the discussion around \cite[Eq.~(1)]{Schuch-Verstraete}, one must understand whether the restrictions of the ground state to widely separated regions $X$ and $Y$, denoted $\psi|_X$ and $\psi|_Y$, approximately factorize:
\begin{equation}\label{eq: exp decay mutual correlations}
\psi|_{X\cup Y} \simeq \psi|_X \otimes \psi|_Y,
\end{equation}
up to terms exponentially small in the distance $\mathrm{dist}(X,Y)$ separating $X$ and $Y$.  It should be noted that this is a strictly stronger property than $\psi(O_XO_Y)\simeq\psi(O_X)\psi(O_Y)$ (again up to exponentially small terms) for any $O_X,O_Y$ supported in $X$, respectively $Y$. The latter property simply follows from having a gap, see \cite{HastingsKoma2006,NachtergaeleSims2006}.
We refer to the property \eqref{eq: exp decay mutual correlations} as decay of mutual correlations, which we establish in Theorem~\ref{corollary.factorizing-infinite-regions}. Our approach, however, is conceptually distinct from those in the works cited above this equation: our methods are rooted in the operator-algebraic classification of gapped phases of matter.

Within that framework, the principal problem is to describe equivalence classes of gapped ground states under quasi-local automorphisms. In one dimension, the relevant classification result asserts that symmetry-protected topological (SPT) phases are classified by the group cohomology $H^2(G,U(1))$, where $G$ is the symmetry group \cite{Chen_2013,quella,guwen2009,schuch2011MatrixProduct, pollmann_entanglement_2010}. In particular, in the absence of symmetry there exists only a single phase. 
The corresponding index in infinite volume was constructed in \cite{Ogata2020_ClassificationSPT_QSC} and, consequently, in \cite{kapustin-classification2021} it was shown that this classification is complete (in \emph{stable homotopy}, i.e.\ allowing for tensoring with product states) when $G$ is a finite group. Extensions beyond finite $G$ are discussed in \cite{quella,de_o_carvalho_classification_2025,Kubota2025}.  However, none of these works addresses directly the classification of gapped ground states, but rather the classification of SRE states. The aim of the present paper is precisely to bridge this conceptual gap. Indeed, we show that all gapped ground states of finite-range interactions are SRE. That being established, the above classification results on SRE states can be lifted to the class of gapped ground states, which is, from a physics point of view,  the relevant class. This goal of classifying  gapped ground states was an important motivation for the present work. 


The present work builds on the techniques developed in \cite{kapustin-classification2021} and the related analysis in \cite{sopenko_chiral_2023}. Our proof proceeds in two main steps. First, using the results of \cite{hastings.arealaw} and \cite{ukai-hastings-infinite-volume}, we establish the decay of mutual correlations. Combining this with a lemma from \cite{sopenko_chiral_2023}, we obtain a local version of the split property: we construct a unitary $U$ such that the state $\psi\circ\mathrm{Ad}(U)$ factorizes as a product between the left and right half-chains. Since this modification does not affect the state far from the cut, repeating the construction at sufficiently separated cuts yields a global product state. Finally, to show that the sequence of local transformations converges to a genuine automorphism, we use a technique from \cite{kapustin-classification2021}.

\section{Setup and Main Result}

 To each site $x \in \ZZ$, we attach a finite-dimensional complex Hilbert space $\HH_x$, and denote by $\AA_x = \mathcal B(\HH_x)$, the algebra of \textit{observables} supported at site $x$. We assume that $\max_x d_x  < \infty$, with $d_x= \dim(\HH_x)$.
For finite $X\subset \ZZ$, we set 
$$\AA_X := \bigotimes\limits_{x \in X} \AA_x$$
The local algebra $\AA^{\text{loc}}$ is obtained as an inductive limit of these algebras $\caA_X$, under a suitable identification
$$A_X \sim A_X \otimes \mathds 1_{Y} \in \AA_{X \cup Y}$$
for finite disjoint sets $X,Y$. The quasi-local $C^*$-algebra $\AA$ is the norm closure of $\AA^{\text{loc}}$ with respect to its associated operator norm topology.  It comes equipped with subalgebras $\caA_Y$ for any (finite or infinite) $Y\subset \bbZ$. For infinite $Y$, $\caA_Y$ is the smallest $C^*$-algebra that contains $\caA_X$ for any finite $X\subset Y$. 
We refer to \cite{bratteliII,simon2014statistical,naaijkens2017quantum} for explicit details of this standard construction.

\subsection{Interactions}

A (discrete) interval $X\subset \ZZ$ is a set of the form $\{ x\in \ZZ, a\leq x \leq b\}$ with $a \in \ZZ \cup \{-\infty\} $ and $b \in \ZZ \cup \{\infty\} $. 
An interaction is a collection $\Phi=(\Phi_I)_{I}$, where $I$ are finite intervals and $\Phi_I \in \caA_I$ with  $\Phi_I=\Phi^*_I$.
We say an interaction is finite-range whenever, for some finite $r$,  
\begin{equation}\label{eq: finite range}
    |I| >r \Rightarrow \Phi_I=0.
\end{equation}
The \emph{range} $r_{\Phi}$  of $\Phi$ is the smallest $r\geq 0$ for which this holds. 
We say the finite-range interaction $\Phi$ is locally bounded iff.\ $\sup_{I}\norm{\Phi_I} <\infty $ where the supremum is over finite intervals.

We say an interaction is exponentially quasi-local whenever there are constants $C <\infty,c>0$ such that
\begin{equation} \label{eq: exponential decay interaction}
||\Phi_I||\leq Ce^{-c |I|}.
\end{equation}

We will also need families of interactions, parametrized by time $t\in [0,1]$.  We say such a family $(\Phi(t))$ is an exponentially quasi-local time-dependent interaction whenever \eqref{eq: exponential decay interaction} holds for $\Phi(t)$, uniformly in $t\in [0,1]$, and, for each finite interval $I$, $t\mapsto \Phi_I(t)$ is continuous.

\subsection{Derivations and LGAs} \label{sec.lga}
For an interaction $\Phi$ (of finite-range or exponentially quasi-local), the following expression is well-defined for any $A \in \AA^{\text{loc}}$:
$$\delta_{\Phi}(A)= \sum\limits_{I } i[\Phi_I,A]$$
where the sum is over finite intervals. 
In general, $\delta_{\Phi}$ is an unbounded map $\caA^{\text{loc}}\to \caA$ that is called a \emph{derivation}. 
If $(\Phi(t))_{t\in [0,1]}$ is an exponentially quasi-local time-dependent interaction, then the family of corresponding derivations generates a time-dependent family of automorphisms $(\alpha_{t})_{t\in [0,1]}$ solving the Heisenberg equation
$$
\frac{d}{dt}\alpha_t(A)=  \alpha_t(\delta_{\Phi(t)}(A)),\qquad  \alpha_0(A)=A,
$$
for $A$ in a dense set of $\caA$.
The well-definedness and locality properties of these $\alpha_t$ are highly non-trivial, see 
\cite{liebrobinson,nachtergaele.sims.ogata.2006}, but very well-understood by now.
We refer to $\alpha_1$ as a 
\emph{locally generated automorphism} (LGA) and we write it simply as $\alpha$.  

\subsection{States}

States are normalized positive linear functionals on a $C^*$-algebra $\AA$. 
The set of states on $\AA$ (denoted by $\mathcal S(\AA)$) is convex and its extremal points are called pure states. For $Y\subseteq \ZZ$, we denote by $\psi|_{Y}$ the restriction of a linear functional $\psi$ to a subalgebra $\mathcal \AA_Y \subset \AA$, characterized by $\psi(A) = \psi|_{Y} (A)$ for every $A \in \mathcal \AA_Y$.

A product state $\phi \in \mathcal S(\caA)$ is a state characterized by 
$$
\phi(AB)=\phi(A) \phi(B)
$$
whenever $A \in \caA_X,B \in \caA_Y$ with $X\cap Y =\emptyset$.
A \textit{short-range entangled} (SRE) state $\psi$ is a state such that, for some LGA $\alpha$ and some pure product state $\phi$, 
$$
\psi=\phi\circ\alpha.
$$
The definition of SRE states differs slightly between different sources, mainly because of changes in the locality requirements of the time-dependent interaction generating the LGA. \\

A \textit{gapped ground state} associated to an interaction $\Phi$ is a state $\psi \in \mathcal S(\AA)$ such that 
\begin{equation} \label{eq.def-gapped-state}
    -i\psi(A^*\delta_{\Phi}(A)) \ge     \gamma \psi(A^*A), \qquad \forall A \in \caA^{loc} \,\, \text{such that} \,\, \psi(A)=0
\end{equation}
for some parameter $\gamma>0$ that is called (a lower bound for) the \emph{gap}. Here is one way to explain this terminology: In a GNS representation $(\HH,\pi,\ket{\psi})$ of a gapped ground state $\psi$, there is a unique nonnegative operator $H_{GNS}$ that implements the dynamics $\alpha_t$, as in 
\begin{align*}
    e^{itH_{GNS}} \pi(A) \ket\psi = \pi(\alpha_t(A)) \ket \psi,
\end{align*}
for all $A \in \AA$ and $t \in \RR$, such that $H_{GNS} \ket \psi = 0$. $0$ is a nondegenerate eigenvalue of $H_{GNS}$ and the parameter $\gamma>0$ appearing on the right-hand side of \eqref{eq.def-gapped-state} is a lower bound for the spectral gap of $H_{GNS}$. Moreover, {a gapped ground state is necessarily pure, see \cite[Theorem A.3]{Tasaki2022}.}


\subsection{Main result}
Referring to the terminology introduced above, we can now formulate our main result.
\begin{theorem} \label{thm.main}
    Let $\Phi$ be a finite-range, locally bounded interaction, as defined above. Any gapped ground state associated to $\Phi$ is SRE.
\end{theorem}

Let us make two short comments on the assumptions. 

\begin{enumerate}
    \item 
The requirement that $\Phi$ has finite-range originates in the fact that the crucial references \cite{matsui.splitproperty} and \cite{hastings.arealaw,ukai-hastings-infinite-volume} assume this finite-range condition. 

\item
The  theorem applies in cases where there are multiple gapped ground states associated to $\Phi$, as it happens in models exhibiting spontaneous symmetry breaking. Consider for example the (classical) Ising interaction for $\Phi$, i.e.\ $\Phi_{I}=-Z_iZ_{i+1}$ (with $Z_i$ the third Pauli matrix at site $i$) for $I=\{i,i+1\}$ with $i\in \Z$ and $\Phi_I=0$ for all other $I$. Some ground states associated to this $\Phi$ are: 
\begin{enumerate}[label=(\roman*)]
    \item  The ``all-spins up" product state $\psi_+$ defined by $\psi_+(Z_i)=1$.
        \item  The ``all-spins down" product state $\psi_-$ defined by $\psi_-(Z_i)=-1$.
        \item The convex combinations $\psi_a= a\psi_++(1-a) \psi_- $ for $0<a<1$.
        \item The domain wall product states $\psi_j$ defined by $\psi_j(Z_i)=\chi(i>j)-\chi(i\leq j)$ and any coherent and incoherent superpositions of these domain wall product ground states.
\end{enumerate}
 The ground states $\psi_+,\psi_-$ satisfy the condition  \eqref{eq.def-gapped-state}. Indeed, the ground state  $\psi_{+}$ and the domain wall ground states are not realizable in the GNS representation of $\psi_{-}$: hence the GNS Hamiltonian corresponding to $\psi_-$ has a non-degenerate ground state. The same holds of course with the roles of $+,-$ reversed.  In these cases, the conclusion of the theorem is trivial since the ground states $\psi_+,\psi_-$ are manifestly product states. 
However, the ground states $\psi_a$ are definitely not SRE states. They are not pure and they have correlations that do not decay in space. The reason that the theorem does not apply for them, is that \eqref{eq.def-gapped-state} fails, which is seen by choosing $A=Z_i-(2a-1)\mathds 1$ for any $i\in \Z$.  The domain wall product states $\psi_j$ are pure and they are SRE states, but still our theorem does not apply because \eqref{eq.def-gapped-state} fails, as can be seen by choosing $A=X_j$.
\end{enumerate}

The rest of the paper is devoted to the proof of Theorem \ref{thm.main}, culminating in Section \ref{sec.short-global-disentangling}.

\section{Preliminaries and notation}\label{sec: prelim and notation}

In all what follows, we fix a finite-range interaction $\Phi$ and a pure gapped ground state $\psi$ associated to it. The symbol $\psi$ is always used with this meaning and we will not repeat this every time.   Constants $C<\infty,c>0$ are quantities that can depend on $\Phi,\psi$ but on nothing else, and that can be updated from line to line. The triple $(\HH,\pi,\ket\psi)$ always refers to a GNS triple of $\psi$. 

In what follows, we will assume for simplicity that the on-site dimensions $d_x$ satisfy $d_x \geq 2$. This incurs no loss of generality\footnote{If there is a consecutive stretch of sites with $d_x=1$, with length longer than $r_\Phi$ (range of the interaction $\Phi$), then the ground state factorizes between the left and right hand sides of this stretch and we can treat the factors separately.  Stretches not longer than $r_\Phi$ are eliminated by grouping sites together. We hence reduce the problem to the case where $d_x\geq 2$.}.



\subsection{Subsets of \texorpdfstring{$\bbZ$}{}}

For a region $X \subset \bbZ$, we define its \textit{$r$-fattening} $X_r$, by
\begin{equation*}
    X_r := \{ y \in \ZZ\ |\ \text{dist}(\{y\},X) \le r \}
\end{equation*}
where $\text{dist}(Y,X) = \inf_{\substack{x \in X, \ y \in Y}} | x - y |$. We use the same notation for a negative fattening parameter: namely, for $r>0$, we let $X_{-r} = [(X^c)_r]^c$, where $X^c = \{y \in \ZZ, y \not\in X\}$ denotes the complement of region $X$. We denote by $[a,b] = \{ x\in \ZZ, a\leq x \leq b\}$ for $a,b \in \ZZ$. Moreover, by $(a,b) = [a+1,b-1]$ for $a,b \in \ZZ$, and by $B_r(x)$ the interval $(x-r,x+r)$ centered at $x \in \ZZ$ with radius $r>0$. When $x=0$, we simplify the notation by writing $B_r := B_r(0)$. We also use the notation $< x$ (and $\ge x$) for the set $\{y \in \ZZ,\ y < x\}$ (respectively, $\{y \in \ZZ,\ y \ge x\}$). \par 
The \textit{$r$-width boundary} of $X$ is defined as
\begin{equation*}
    \partial X(r) :=  X_r \cap (X^c)_r. 
\end{equation*}


\subsection{Localization of states and observables} \label{sec.localization}

We say two states $\varphi, \varphi' \in \mathcal S(\AA)$ are \textit{(K,k)-exponentially close far from a site $x$} if there are constants $K<\infty,k>0$, such that
\begin{equation}
    |\varphi(A) -\varphi'(A)| \le Ke^{-k\ell} \norm{A}
\end{equation}
for any $A \in \AA_{B_\ell(x)^c}$.

An element $A \in \AA$ is said to be \textit{(K,k)-exponentially anchored in a finite set $X \subset \ZZ$} if there are constants $K<\infty, k>0$, such that
\begin{align*}
    \norm{[A,B]} \le Ke^{-k \cdot \dist(X,Y)} \norm{A} \norm{B} ,
\end{align*}
for every $Y\subset \ZZ$ and every $B \in \AA_Y$.
The following approximation result for exponentially anchored unitaries is often convenient.
\begin{lemma} \label{lem.local-approximation-unitary}
    Let $U\in\AA$ be a unitary $(K,k)$-exponentially anchored at a
    nonempty finite set $X\subset\ZZ$. For every integer $\ell\geq0$,
    there is a unitary $U_\ell\in\AA_{X_\ell}$ such that
    \begin{equation}\label{eq.anchored-unitary-approximants}
        \norm{U-U_\ell}\leq 2Ke^{-k\ell}.
    \end{equation}
    Moreover, each $U_\ell$ is $(5K,k)$-exponentially anchored at $X$. 
\end{lemma}
Lemma \ref{lem.local-approximation-unitary} 
is standard but we sketch the proof in Appendix~\ref{app.anchored-unitaries}.  The following lemma is a rather straightforward consequence, and we omit its proof.
\begin{lemma}\label{lem.anchored-unitary-lga}
    Let $U\in\AA$ be a unitary $(K,k)$-exponentially anchored at a
    site $x\in\ZZ$. Then $\Ad{U}$ is an LGA generated by an  exponentially quasi-local time-dependent interaction
    $\Phi(t)$, for which $\Phi_I(t)=0$ unless
    $I=(\{x\})_\ell$ for some integer $\ell\geq0$, and
    \begin{equation}\label{eq.anchored-lga-interaction}
        \sup_{t\in[0,1]}\norm{\Phi_{(\{x\})_\ell}(t)}
        \leq K'e^{-k\ell}.
    \end{equation}
    Here $K'$ depends only on $K$ and $k$.
\end{lemma}

\subsection{Factorization of the GNS representation}

The following Lemma is a consequence of \cite[Lemma 2.7]{ukai-hastings-infinite-volume}. Notice that, for the moment, the finiteness assumption on the interval $I$ is necessary. When the interval $I$ is infinite, the result is known as the \emph{split property}, which is the content of Corollary \ref{cor.split-property}. 

\begin{lemma} \label{lem.split-finite}
    Let $I$ be a finite interval. There are irreducible representations $(\HH_I, \pi_I)$ and $(\HH_{I^c}, \pi_{I^c})$ of $\AA_I$ and $\AA_{I^c}$, respectively, and a cyclic vector $\ket{\psi}$ such that
    \begin{align*}
        (\HH_I \otimes \HH_{I^c}, \pi_I \otimes \pi_{I^c},\ket{\psi})
    \end{align*}
    is a GNS triple for $\psi$. Moreover, $\mathcal B(\HH_\Gamma) = \pi_\Gamma(\AA_\Gamma)''$ for $\Gamma = I,I^c$. 
\end{lemma}

\subsection{Norms on states} \label{sec.norms-states}




The Banach space norm on linear functionals on $\caA$ 
is often used as a metric between states $\varphi,\varphi'$, i.e., 
\begin{equation}\label{eq: metric on states}
 ||\varphi-\varphi'|| = \sup_{A \in\caA, ||A||=1} |\varphi(A)-\varphi'(A)|.
 \end{equation}
We will almost exclusively use functionals $\zeta$ on $\caA$ that are normal w.r.t.\  a given pure state $\psi$, i.e.\ they can be represented in the form
$$
\zeta(A) = \Tr_{\HH} (\dm{\zeta} \pi(A)),
$$
where $\dm\zeta \in \caB_1(\caH)$, i.e.\ $\dm\zeta$ is a trace-class operator on $\caH$.  If $\zeta$ is a state, then $\dm\zeta$ is a density matrix. The important point here is that, for normal functionals, the trace-class norm
\begin{equation}
||\dm\zeta||_1=\Tr_{\HH}(\sqrt{\dm\zeta^*\dm\zeta})
\end{equation}
coincides with its norm as a functional. Indeed, 
\begin{align} 
    \nonumber \norm{\dm\zeta}_1  &= \sup\limits_{\substack{B \in \pi(\AA)'' \\ ||B||=1}} |\Tr_{\HH}(\dm\zeta B)|  = \sup\limits_{\substack{B \in \pi(\AA) \\ ||B||=1}} |\Tr_{\HH}(\dm\zeta B)| = \sup_{\substack{A \in \caA \\ ||A||= 1}} |\zeta(A)|.
\end{align}
The first equality follows from $\caB(\caH)=\pi(\caA)''$. The second equality uses that $\pi(\AA)''$ is the closure of $\pi(\AA)$ in the weak or strong operator topology, by von Neumann's bicommutant theorem, and the Kaplansky density theorem.


When $\dm\zeta$ is pure, we frequently choose a normalized representative in the range of $\dm\zeta$, and denote it by $\ket \zeta$. Conversely, a unit vector $\ket \zeta$ on a given irreducible representation of $\AA_Y$, $Y \subseteq \ZZ$, can be lifted to a pure state on $\AA_Y$ in a standard way. 

Moreover, suppose we are given a factorized GNS representation $(\HH_{Y} \otimes \HH_{Y^c}, \pi_Y \otimes \pi_{Y^c})$ of $\psi$ such that $\mathcal B(\HH_\Gamma) = \pi_\Gamma(\AA_\Gamma)''$, for $\Gamma = Y,Y^c$. For a normal functional $\zeta$, a similar reasoning yields
\begin{align} \label{eq.restricted-density-matrix}
    \norm{\Tr_{\HH_{Y^c}}(\dm{\zeta})}_1 = \sup_{ \substack{ B \in \pi_{Y}(\AA_{Y})''\\\norm{B}=1 }} |\Tr_{\HH}(\dm\zeta B) | = \sup_{\substack{B \in \pi_{Y}(\AA_Y) \\ \norm{B}=1}} |\Tr_\HH (\dm \zeta B)| = \norm{\zeta|_{Y}},
\end{align}
where $\Tr_{\HH_{Y^c}}$ denotes the partial trace over $\HH_{Y^c}$. 




\section{Exponential decay of mutual correlations} \label{sec.exponential-decay-mutual-correlations}
We prove the exponential decay of mutual correlations, already mentioned in the introduction. The main result is Theorem \ref{corollary.factorizing-infinite-regions}, which appears as a simple corollary of Lemma \ref{thm.mutual-correlation-decay}. 

For the duration of Section \ref{sec.exponential-decay-mutual-correlations}, given a finite interval $I$, we choose the GNS representation of $\psi$ to be factorized with respect to the bipartition $I\cup I^c = \ZZ$, e.g.\ Lemma \ref{lem.split-finite}.

\subsection{Hastings Factorization}

The following theorem is a compilation of Corollary~3.2, Corollary~3.9, and Corollary~3.13 in \cite{ukai-hastings-infinite-volume}, which, in turn, is based on \cite{hastings.arealaw}, see also \cite{hamza}. It is the main tool in the present section \ref{sec.exponential-decay-mutual-correlations}.
\begin{theorem} [Hastings factorization in infinite volume, \cite{ukai-hastings-infinite-volume}] \label{thm.hastings-factorization}
There are constants $C< \infty,c>0$ such that, for any finite interval $I $ and any $\ell\in \bbN$, there are projections $O_{I^c}(\ell) \in \mathcal B(\HH_{I^c})$, $O_{I}(\ell) \in \pi(\AA_I)$, and a positive contraction $O_{\partial I}(\ell) \in \pi(\AA_{\partial I(\ell)})$ such that
\begin{equation} \label{eq.approximation-to-projection}
    \norm{O_{\partial I}(\ell)O_I (\ell) O_{I^c}(\ell)  - \kb\psi} \le C e^{-c\ell}.
\end{equation}
and 
\begin{equation} \label{eq.factorization-almost-identity}
    \norm{(O_\Gamma(\ell) - \mathds 1) \ket\psi} \le C e^{-c\ell}, \qquad \text{for }\Gamma \in \{I,I^c,\partial I\}. 
\end{equation}
\end{theorem}

\subsection{Schmidt decomposition} \label{sec.schmidt}

There exists a Schmidt decomposition of the GNS representative
\begin{equation} \label{eq.schmidt-decomposition}
    \ket\psi = \sum\limits_{j=1}^{\sr{I}} \sqrt{\lambda_j} \ket{\psi_j^I} \otimes \ket{\psi_j^{I^c}},
\end{equation}
with ordered strictly positive eigenvalues $\lambda_1 \ge \lambda_2 \ge \dots \ge \lambda_{\text{sr(I)}}$ and 
 $\text{sr(I)}$ the Schmidt rank of $\ket \psi$ with respect to the bipartition $I \cup I^c = \ZZ$. Note that $\text{sr(I)}\leq \mathrm{dim}(\caH_I) <\infty$.
We also note that the Schmidt spectrum is the same regardless the choice of factorized (between $I$ and $I^c$) GNS-representation. 
Our conventions for the generic constants  $C,c$ in particular imply that they do not depend on the interval $I$. This convention is used throughout the paper without further comment.

\begin{proposition} \label{prop.decay-of-schmidt-coeffs}
    There exists a constant $\alpha>0$ such that, for any finite interval $I$, 
    \begin{equation} \label{eq.decay-schmidt-coefficients}
        \sum\limits_{j \ge \cutoff} \lambda_j \le C\cutoff^{-\alpha}, \qquad \text{for any $\cutoff \ge 1$.}
    \end{equation}

\end{proposition}
\begin{proof}
    This is a consequence of Theorem \ref{thm.hastings-factorization}, see equation (29) of \cite{hastings.arealaw}. The proof in \cite{hastings.arealaw} is in finite volume, but it can be extended \cite{matsui.splitproperty} to the case of infinite volume. 
\end{proof}


Given a Schmidt decomposition as in \eqref{eq.schmidt-decomposition}, we define
\begin{align} \label{eq.off-diagonal-X-comp}
    &\dm{\mathfrak I^c_{m,n}} := \ket{\psi_m^{I^c}} \bra{\psi_n^{I^c}} \in \mathcal B_1(\HH_{I^c}),
\end{align}
and
\begin{align} \label{eq.off-diagonal-X}
    &\dm{\mathfrak I_{m,n}} := \ket{\psi_m^{I}} \bra{\psi_n^{I}} \in \mathcal B_1(\HH_{I}).
\end{align}
The trace-class operators defined in \eqref{eq.off-diagonal-X-comp} and \eqref{eq.off-diagonal-X} induce linear functionals $\mathfrak I^c_{m,n}$ and $\mathfrak I_{m,n}$ on $\AA_{I^c}$ and $\AA_I$, respectively. 


\begin{lemma} The following items hold: \label{lem.hastings-decay-of-correlations}

\begin{enumerate}    
    \item For $A \in \AA_{I_{-\ell}}$, 
    \begin{equation}
        \left| \expval{\pi (A) \dm{\mathfrak I^c_{m,n}}}{\psi} - \psi(A) \expval{\dm{\mathfrak I^c_{m,n}}} {\psi} \right| \le  Ce^{-c\ell}\norm{A}.
    \end{equation}

    \item For $A \in \AA_{(I^c)_{-\ell}}$, 
    \begin{equation}
        \left| \expval{\pi(A) \dm{\mathfrak I_{m,n}}}{\psi} - \psi(A) \expval{ \dm{\mathfrak I_{m,n}}} {\psi} \right| \le Ce^{-c\ell}\norm{A} .
    \end{equation}
\end{enumerate}
\end{lemma}

\begin{proof}
To prove item 1, we must repeatedly use Hastings decomposition operators $ O_\Gamma$ for $\Gamma\in \{I,I^c,\partial I\}$. Firstly, we use \eqref{eq.factorization-almost-identity} to obtain
    \begin{align*}
        \left|  \expval{\pi(A) \dm {\mathfrak I^c_{m,n}}}{\psi} - \expval{O_{\partial I}O_{I^c} \pi(A) \dm {\mathfrak I^c_{m,n}} O_I}{\psi}   \right| \le 3Ce^{-c\ell}\norm{A}. 
    \end{align*}
    Since $A \in \AA_{I_{-\ell}}$, $[\pi (A),O_{\partial I}] = 0$, and $[\pi (A),O_{I^c}] = 0$.
   Hence,
    \begin{align} \label{eq.lemma-mutualdecay-1}
        \left|  \expval{\pi (A)\dm{\mathfrak I^c_{m,n}}} {\psi} - \expval{\pi (A) O_{\partial I} O_{I^c} O_I  \dm{\mathfrak I^c_{m,n}}}{ \psi}  \right| \le 3Ce^{-c\ell}\norm{A}. 
    \end{align}
    The result then follows from the factorization property \eqref{eq.approximation-to-projection}, by noting that $[O_I,O_{I^c}]=0$, and by updating the constants. The proof of item 2 is analogous, and follows by replacing the roles of $I$ and $I^c$ above. 
\end{proof}

\begin{proposition} \label{prop.closeness-of-density-matrices}

    Given a Schmidt decomposition \eqref{eq.schmidt-decomposition}, 
\begin{equation*}
        \norm{ \mathfrak I_{m,n}|_{{I_{-\ell}}} -\delta_{n,m} \psi|_{{I_{-\ell}}}  } \le \frac{1}{\sqrt{\lambda_m\lambda_n}} Ce^{-c\ell}, \qquad \text{and} \qquad   \norm{ \mathfrak I^c_{m,n}|_{{(I^c)_{-\ell}}} 
         -\delta_{n,m} \psi|_{{(I^c)_{-\ell}}} 
        } \le \frac{1}{\sqrt{\lambda_m\lambda_n}} Ce^{-c\ell},
    \end{equation*}
for any $n,m \le \sr{I}$.

    
\end{proposition}
\begin{proof} 
For $A \in \caA_{I}$, the definition of Schmidt decomposition yields
 $$  \expval{ \pi (A)  \dm {\mathfrak I^c_{m,n} }}{\psi}  =  \sqrt{\lambda_m \lambda_n}  \bra {\psi_m^I} \pi (A) \ket {\psi_n^I}.
 $$
Therefore the first bound follows from item 1 of Lemma \ref{lem.hastings-decay-of-correlations}. The second bound  follows in a similar way from item 2 of Lemma \ref{lem.hastings-decay-of-correlations}, exchanging the roles of $I$ and $I^c$.

\end{proof}

\subsection{Exponential decay of mutual correlations} \label{sec.mutual-correlations}

We say a state $\varphi \in \mathcal S(\AA)$ satisfies \textit{$(K,k)$-exponential decay of mutual correlations} if, for every interval $Y$ (possibly infinite), it holds that
\begin{equation} \label{eq.def-mutual-exp-decay}
        \norm{\varphi|_{Y_{-\ell} \cup (Y^c)_{-\ell}} - \varphi|_{Y_{-\ell}} \otimes \varphi|_{(Y^c)_{-\ell}}} \le Ke^{-k\ell},
    \end{equation}
for each $\ell > 0$. 

 In preparation for Theorem \ref{corollary.factorizing-infinite-regions}, which establishes exponential decay of mutual correlations for $\psi$, we state the following related result.  It differs from \eqref{eq.def-mutual-exp-decay} in that we now only consider finite intervals. 
\begin{lemma} \label{thm.mutual-correlation-decay}
 For every finite interval $I$, and all $\ell>0$, it holds that 
    \begin{equation}
        \norm{\psi|_{I_{-\ell}\cup (I^c)_{-\ell}} - \psi|_{I_{-\ell}} \otimes \psi|_{(I^c)_{-\ell}}} \le Ce^{-c\ell}.
    \end{equation}
\end{lemma}

\begin{proof}
    We consider a Schmidt decomposition as in \eqref{eq.schmidt-decomposition}, and define the pure state truncations
    \begin{equation}
        \ket{\psi^{[\cutoff]}} := \dfrac 1{\mathfrak n} \sum\limits_{j = 1}^{\cutoff} \sqrt{\lambda_j}\ \ket{\psi_j^{I}} \otimes \ket{\psi_j^{I^c}}
    \end{equation}
    parametrized by $\cutoff \in \NN$, where $\mathfrak n^2=\sum_{j=1}^{\cutoff} \lambda_j$. 
    {For $j>\mathrm{sr}(I)$, we set $\lambda_j=0$.}
    By Proposition \ref{prop.decay-of-schmidt-coeffs}, it follows that
    $1 \ge \mathfrak n \ge 1-C(\cutoff)^{-\alpha} $ and
    \begin{align*}
        \norm{\ket \psi - \ket {\psi^{[\cutoff]}}}^2 &= 
        \sum\limits_{j=1}^{\cutoff} (1-1/\mathfrak n)^2\lambda_j + \sum\limits_{j > \cutoff} \lambda_j \le C(\cutoff)^{-\alpha}.
    \end{align*}
    It follows that
    \begin{equation}\label{eq: truncation bound states}
     \norm{\psi - \psi^{[\cutoff]}}  \le 2\norm{\ket \psi - \ket{\psi^{[\cutoff]}}} \le  C(\cutoff)^{-\alpha/2}. 
    \end{equation}
    We proceed by bounding
    \begin{align}
        \nonumber \norm{\psi^{[\cutoff]}|_{I_{-\ell} \cup (I^c)_{-\ell}} - \psi|_{I_{-\ell}} \otimes \psi|_{(I^c)_{-\ell}}}  &\le \BigglVert \sum\limits_{m,n=1}^{\cutoff} \frac 1{\fk n^2} \sqrt{\lambda_m\lambda_n} (\mathfrak I_{m,n}\otimes \mathfrak I^c_{m,n})|_{I_{-\ell} \cup (I^c)_{-\ell}}- \\ \nonumber &- \sum\limits_{n=1}^{\cutoff} \lambda_n \psi|_{I_{-\ell}} \otimes \psi|_{(I^c)_{-\ell}}\BiggrVert         
        \nonumber + \norm{\sum\limits_{n=\cutoff+1}^\infty \lambda_n \psi|_{I_{-\ell}} \otimes \psi|_{(I^c)_{-\ell}}} \\ 
        \label{eq.mutual-correlation-bound1}&\le \sum\limits_{n=1}^{\cutoff} \lambda_n\norm{\frac{1}{\fk n^2} \mathfrak I_{n,n}|_{I_{-\ell}} \otimes \mathfrak I^c_{n,n}|_{(I^c)_{-\ell}}  - \psi|_{I_{-\ell}} \otimes \psi|_{(I^c)_{-\ell}} } \\
        &\label{eq.mutual-correlation-bound2}+ \sum\limits_{\substack{m,n=1 \\ m\neq n} }^{\cutoff} \frac 1{\fk n^2} \sqrt{\lambda_m\lambda_n} \norm{\mathfrak I_{m,n}|_{I_{-\ell}} \otimes \mathfrak I^c_{m,n}|_{(I^c)_{-\ell}}} \\
        &\label{eq.mutual-correlation-bound3}+ \sum\limits_{n=\cutoff+1}^{\infty} \lambda_n.
    \end{align}

\begin{enumerate}
    \item The first term \eqref{eq.mutual-correlation-bound1} is upper bounded by 
    \begin{align*}
        &\sum\limits_{n=1}^{\cutoff} \lambda_n \left( \frac{1}{\fk n^2} - 1  + \norm{\mathfrak I_{n,n}|_{I_{-\ell}} - \psi|_{I_{-\ell}} } + \norm{\mathfrak I^c_{n,n}|_{(I^c)_{-\ell}} - \psi|_{(I^c)_{-\ell}}} \right) \\
        &\le \left(\frac{1}{\fk n^2} - 1   \right) + \sum\limits_{n=1}^{\min(\sr I,\cutoff)} \lambda_n \left( \frac 2 {\lambda_n} C e^{-c\ell} \right) \\
        &\le C (\cutoff)^{-\alpha }+ 2\cutoff Ce^{-c\ell}. 
    \end{align*}
    by  Proposition \ref{prop.closeness-of-density-matrices}

 \item
The second term \eqref{eq.mutual-correlation-bound2} includes off-diagonal linear functionals, and it is upper bounded by 
\begin{align*}
    \frac 1{\fk n^2} \cutoff(\cutoff-1) Ce^{-c\ell}
\end{align*}
by  Proposition \ref{prop.closeness-of-density-matrices}.

\item Finally, the third term \eqref{eq.mutual-correlation-bound3} can be upper bounded by $C(\cutoff)^{-\alpha}$ according to the decay \eqref{eq.decay-schmidt-coefficients}. 
\end{enumerate}

\noindent {Hence we conclude, combining \eqref{eq: truncation bound states}, and the bound by (\ref{eq.mutual-correlation-bound1}, \ref{eq.mutual-correlation-bound2}, \ref{eq.mutual-correlation-bound3}), that
\begin{align*} 
    \norm{\psi|_{I_{-\ell}\cup (I^c)_{-\ell}} - \psi|_{I_{-\ell}}\otimes \psi|_{(I^c)_{-\ell}}}
    &\le   C(\cutoff)^{-\alpha/2}+ C (\cutoff)^2 e^{-c\ell} + C(\cutoff)^{-\alpha}.
\end{align*}}

Let us choose $\cutoff = \lceil{e^{b\ell}}\rceil$ for some $0<b<c/2$ with $c$ the constant in the first term in the above line,  then the above sum is bounded by $Ce^{-c \ell}$ (after updating the value of $c$).

\end{proof}


Now we are finally ready to state and prove the actual
\begin{theorem} [Exponential decay of mutual correlations] \label{corollary.factorizing-infinite-regions}
    The gapped ground state $\psi$ satisfies $(C,c)$-exponential decay of mutual correlations.
\end{theorem}
\begin{proof}
{It only remains to show
\eqref{eq.def-mutual-exp-decay} for infinite intervals $Y$, that is, for $Y$ being the halfline $\{\cdot \leq x\}$. }
By density, it suffices to prove 
\begin{equation}\label{eq: to show for mutual decay}
     \left| \psi(A) - \psi|_{\le x - \ell} \otimes \psi|_{\ge x+\ell+1} (A) \right| \leq \norm{A} Ce^{-c\ell},
\end{equation}
for  $A \in \AA^{\text{loc}}_{B_\ell(x)^c}$, and with $c,C$ uniform in $A$. 
Let $a>\ell $ be such that $A \in \AA_{[x-a,x-\ell]} \otimes \AA_{[x+\ell+1,x+a]}$. We then choose $I = [x+1,x+a+\ell]$, 
and \eqref{eq: to show for mutual decay} now follows from Lemma \ref{thm.mutual-correlation-decay}. 
    
\end{proof}

\subsection{Split property}\label{sec: split property}

It is known that a pure gapped ground state $\psi$ associated to a finite range interaction satisfies the \emph{split property} \cite{Matsui2001,matsui.splitproperty}. We state it below and present a short proof, which in the present case follows from Theorem \ref{corollary.factorizing-infinite-regions}.

\begin{corollary} [Split Property] \label{cor.split-property}
    Let $Y$ be an interval (possibly infinite). 
    \begin{enumerate}
        \item 
    The states $\psi$ and $\psi|_Y \otimes \psi|_{Y^c}$ are quasi-equivalent. 


    \item There are irreducible representations $(\HH_Y, \pi_Y)$ and $(\HH_{Y^c}, \pi_{Y^c})$ of $\AA_Y$ and $\AA_{Y^c}$, respectively, and a cyclic vector $\ket \psi$, such that 
        $$(\HH_Y \otimes \HH_{Y^c}, \pi_Y \otimes \pi_{Y^c}, \ket \psi)$$
        is a GNS triple for $\psi$. Moreover, 
$$
\caB(\HH_Y) = \pi(\AA_Y)'', \qquad  \caB(\HH_{Y^c}) = \pi(\AA_{Y^c})''. 
$$
    \end{enumerate}
\end{corollary}
\begin{proof}
    Items 1. and 2. above are equivalent \cite{ogata.completeness}. If $Y$ is a finite interval, we refer to Lemma \ref{lem.split-finite}. Hence, without loss, we prove item 1. for the case $Y=(-\infty, 0]$. 
  Since $\psi$ is pure, it is a factor state. The state $\psi|_Y \otimes \psi|_{Y^c}$ is then also a factor state (for example, because it inherits from $\psi$  condition 2 in \cite[Theorem 2.6.10]{bratteliI}). Finally, since $\psi$ and $\psi|_Y \otimes \psi|_{Y^c}$ are (C,c)-exponentially close by Theorem \ref{corollary.factorizing-infinite-regions} far from $0$, they are quasi-equivalent by \cite[Corollary 2.6.11]{bratteliI}.
\end{proof}




One can somehow strengthen the above decomposition. Given a partition of $\bbZ$ into intervals (some of them infinite) $Y_1,\ldots,Y_n$, then 
\begin{equation} \label{eq.split-many-intervals}
(\caH_{Y_1}\otimes \ldots \otimes \caH_{Y^n}, \pi_{Y_1}\otimes \ldots \pi_{Y_n},\ket \psi)
\end{equation}
is also a GNS-triple for $\psi$. Each $(\HH_{Y_j},\pi_{Y_j})$ is irreducible, such that $\mathcal B(\HH_{Y_j}) = \pi_{Y_j} (\AA_{Y_j})''$.

To see this, one can start from the representation given in Corollary \ref{cor.split-property} with $Y,Y^c$ both half-lines and partition $Y$ and $Y^c$ further. Indeed, if $J$ is a finite interval, then the equivalence of $(\caH_Y,\pi_Y)$ and $ (\caH_J \otimes \caH_{Y\setminus J},\pi_J \otimes \pi_{Y\setminus J})  $ follows by Lemma \ref{lem.split-finite}.

\section{Cutting states}

\noindent
We define the notion  of being \emph{locally SRE}, as a last intermediate step.

We say that a state $\varphi \in \mathcal {S} (\AA)$ is $(K,k)$-locally SRE if, for any $x \in \ZZ$, there is a unitary $U^{(x)} \in \caA$ that is $(K,k)$-exponentially anchored at $x$, such that 
    \begin{equation}
        \varphi = \varphi_{\le x} \otimes \varphi_{>x} \circ \Ad{U^{(x)}}. 
    \end{equation}
    for some $\varphi_{\le x} \in \mathcal {S} (\AA_{\le x})$, $\varphi_{>x} \in \mathcal{S} (\AA_{>x})$.

In this section, we use the properties discussed in section \ref{sec.mutual-correlations} to establish that 
\begin{proposition} \label{prop.cutting}
The state $\psi$ is $(C,c)$-locally SRE.
\end{proposition}

The next section \ref{sec.short-global-disentangling}  will then lift the \emph{local SRE} property to the SRE property that we are ultimately after. 

Our main task is to construct the factorized pure state $\psi_{\le x} \otimes \psi_{>x}$ that is normal w.r.t.\ $\psi$. Once we know that such a state exists and is exponentially close to $\psi$ far from the origin, we can invoke Theorem \ref{thm.connecting-states} from \cite{sopenko_chiral_2023}. 

\subsection{Constructing \texorpdfstring{$\psi_{\le x} \otimes \psi_{>x}$}{}} \label{sec.constructing}

 We start from the factorized mixed state 
 $ \psi|_{\le x} \otimes \psi|_{> x}$, which, by the split property, is normal with respect to $\psi$ (hence can be identified with a density matrix, see Section \ref{sec.norms-states}).  We need to find now a state with the same nice properties as $ \psi|_{\le x} \otimes \psi|_{> x}$, but which is also pure. 

For the duration of section \ref{sec.constructing}, we assume the GNS representation to be factorized between $x$ and $x+1$, but also between $x-\ell$ and $x-\ell+1$, and between $x+\ell-1$ and $x+\ell$, as discussed in section \ref{sec: split property}. Let $\mu_1\geq \mu_2 \geq\ldots $ be the eigenvalues of the density matrix 
$\dm{\psi|_{\leq x} \otimes \psi|_{> x}} \in \mathcal B_1(\HH)$. 

\begin{lemma} \label{lem.bound-eigenvalue-restrictions}
    There is a lower bound
    \begin{equation}
       \mu_1 > c> 0.
    \end{equation}
\end{lemma}

\begin{proof}

We denote by 
$$
 \{\lambda_j({\ell})\}_j, \qquad     \{\mu_j({\ell})\}_j
$$
the eigenvalues of $\Tr_{\HH_{B_\ell(x)}}(\dm{\psi})$ and of $\Tr_{\HH_{B_\ell(x)}}(\dm{\psi|_{\le x} \otimes \psi|_{> x }})$.
Note that we have encountered the eigenvalues $\lambda_j(\ell)$ already in Section \ref{sec.schmidt}, when $I=B_\ell(x)$. We have
\begin{align} 
     \nonumber \sum\limits_{j=1}^\infty |\lambda_j(\ell) - \mu_j(\ell)| &\le \norm{ \Tr_{\HH_{B_\ell(x)}} (\dm{\psi}) - \Tr_{\HH_{B_\ell(x)}} (\dm{\psi|_{\le x} \otimes \psi|_{> x}})}_1 
    \\
    \label{eq: bound on difference eigenvalues} &= \norm{{\psi|_{B_\ell(x)^c}} - { ({\psi|_{\le x} \otimes \psi|_{> x})|_{B^c_\ell(x)} }}} \leq  Ce^{-c\ell},
\end{align}
where we have used the Hoffmann-Wielandt's inequality (equation 5.12 of \cite{Markus_1964}), \eqref{eq.restricted-density-matrix} for the second equality, and Theorem \ref{corollary.factorizing-infinite-regions} for the last inequality.
By a corollary of Ky Fan's Maximum Principle (Lemma 4.2 of \cite{DAFTUAR200580}), it holds that, for any $\ell>0$, 
\begin{equation} \label{eq: ky fan}
    \mu_1 \ge \dfrac{\mu_1(\ell) }{\mathrm{dim}(\caH_{B_\ell(x)})}. 
\end{equation}
Since, by
Proposition \ref{prop.decay-of-schmidt-coeffs}, {we have 
$ \lambda_1(\ell) >c>0$ uniformly in $\ell$, 
we can choose $\ell$ large enough in \eqref{eq: bound on difference eigenvalues}  so that $\mu_1(\ell)\geq \lambda_1(\ell)/2$. Then, we get the claim of the lemma from  \eqref{eq: ky fan}}, by further using that the on-site dimensions are uniformly upper bounded. 
\end{proof}

Since 
$$\dm{\psi|_{\le x} \otimes \psi|_{>x}} = \dm{\psi|_{\le x}} \otimes \dm{\psi|_{>x}} \in \mathcal B_1(\HH_{\le x}) \otimes \mathcal B_1(\HH_{>x})$$
is factorized, we can choose a factorized eigenstate corresponding to its largest eigenvalue $\mu_1$, and we denote it by $|\psi_{\le x} \rangle \otimes |\psi_{> x}\rangle \in \caH_{\le x} \otimes \caH_{> x}$.
It turns out that this inherits the exponential closeness to $\psi$: 
\begin{lemma}\label{prop: exp close factorized pure}
    The pure state $\psi_{\le x} \otimes \psi_{>x}$ is  (C,c)-exponentially close to $\psi$ far from $x \in \ZZ$.
\end{lemma}
\begin{proof}
Theorem \ref{corollary.factorizing-infinite-regions} states that ${\psi}$ and ${\psi|_{\le x} \otimes \psi|_{> x} }$ are (C,c)-exponentially close far from $x$, i.e.
\begin{align*}
\norm{ \Tr_{\HH_{B_\ell(x)}} (\dm{\psi}) - \Tr_{\HH_{B_\ell(x)}} (\dm{\psi|_{\le x} \otimes \psi|_{> x}})}_1 \le  Ce^{-c\ell}.
\end{align*}
Therefore, we can invoke 
  Proposition \ref{prop.krauss}, which yields a completely positive trace preserving map $\Phi (\cdot) = \sum_{j} M_j (\cdot ) M_j^*$, with $M_j \in \caB(\caH_{{B_\ell(x)}})$ and $||M_j||\leq 1$, such that
    $$\norm{\Phi(\dm{\psi}) - \dm{\psi|_{\le x} \otimes \psi|_{>x}}}_1 \le Ce^{-c\ell}$$
   and such that $\Phi(\dm{\psi}) $ has the spectral decomposition 
   $$\Phi(\dm{\psi}) = \sum\limits_{j} \nu_j \kb{\xi_j}, \qquad   \nu_j= \norm{M_j\ket{\psi}}^2, \qquad  \ket{\xi_j}=  \frac 1 {\sqrt{\nu_j}} M_j\ket{\psi}.   $$ 
with eigenvalues $\nu_1 \geq \nu_2 \geq \ldots$ We now apply Lemma \ref{lem: two density matrices} with $\omega = \dm{\psi|_{\le x} \otimes \psi|_{>x}}$ and $\sigma=\Phi(\dm{\psi})$, and with the eigenvector $|\eta\rangle= \ket{\psi_{\le x} \otimes \psi_{>x}}$ and $\lambda=\mu_1$. We conclude that 
$$
\norm{\ket{\psi_{\le x} \otimes \psi_{>x}}- \sum_{j=1}^{N} d_j \ket{\xi_j} } \leq Ce^{-c\ell}
$$
for some coefficients $d_j$ with $\sum_{j=1}^N |d_j|^2 = 1$, $N\leq C$, and such that the corresponding eigenvalues $(\nu_j)_{j=1}^N$ satisfy $\nu_j \geq c>0.$


    
    Now take $A \in \AA_{B_{2\ell}(x)^c}$.  Then 
 \begin{equation} \label{eq: phi and ms}
        \left| \expval{\pi(A)} {\psi_{\le x} \otimes \psi_{>x}} - \sum\limits_{i,j=1}^{N} \bar{d_i} d_j\frac {1}{\sqrt{\nu_i\nu_j}} \expval{M_i^* \pi(A) M_j} {\psi} \right| \leq Ce^{-c\ell} ||A||
      \end{equation}  
and, by Lemma \ref{thm.mutual-correlation-decay}, 

\begin{equation}
| \expval{M_i^* \pi(A) M_j}{\psi} -    \delta_{i,j} \nu_j \psi(A)| \leq Ce^{-c\ell} \norm{A} ||M_i^*|| ||M_j||
\end{equation} 
where we used $\langle\psi|M_i^* M_j|\psi\rangle=\delta_{i,j}\nu_j $. Performing the sum over $i,j$ in \eqref{eq: phi and ms} and using $N\leq C$, $||M_j||\leq 1$ and $\nu_j>c$, we finish the proof.

\end{proof}

\subsection{Proof of Proposition \ref{prop.cutting}}

We have already obtained a pure product state $\psi_{\le x} \otimes \psi_{>x}$ that is (C,c)-exponentially close to $\psi$. {Since these states are mutually normal and pure, by the Kadison transitivity theorem, they are unitarily equivalent.\ The following Theorem, proved in \cite{sopenko_chiral_2023}, allows us to obtain a unitary that is moreover exponentially anchored at $x$, by exploring the locality properties of the states.} We review the proof of Theorem \ref{thm.connecting-states} in Appendix \ref{appendix.sopenko-lemma}.

\begin{theorem} \label{thm.connecting-states}
    Let $\varphi\in\mathcal S(\AA)$ be a pure state satisfying, {for some $K<\infty$ and $k,\alpha >0$,}
    \begin{enumerate}
        \item (K,k)-exponential decay of mutual correlations. 
        \item  Decay of Schmidt eigenvalues for any finite interval $I$ (as defined in Section \ref{sec.schmidt} and in Proposition \ref{prop.decay-of-schmidt-coeffs}):
\begin{align}\label{eq.decay-schmidt-arbitrary-state}
            \sum\limits_{j\ge \cutoff} \lambda_j \le K\cutoff^{-\alpha}, \qquad \cutoff \ge 1.
        \end{align}
    \end{enumerate}
    Furthermore, let $\varphi' \in \mathcal S(\AA)$ be a pure state that is (K,k)-exponentially close to $\varphi$ far from a site $x$.  
    Then there exists a unitary $U \in \AA$, (K',k')-exponentially anchored at $x$, such that 
    \begin{align*}
        \varphi' = \varphi \circ \Ad{U}. 
    \end{align*}
   {The constants $K', k'$ depend only on $K,k,$ and $\alpha$.}
\end{theorem}

\begin{proof}[Proof of Proposition \ref{prop.cutting}] 
The claim follows by applying Theorem \ref{thm.connecting-states} to $\varphi = \psi$ and $\varphi'=\psi_{\le x} \otimes \psi_{>x}$. The conditions 1.,2. are verified by Theorem \ref{corollary.factorizing-infinite-regions} and Proposition \ref{prop.decay-of-schmidt-coeffs}, respectively.  The exponential closeness of $\psi$ and $\psi_{\le x} \otimes \psi_{>x}$ is established in Lemma \ref{prop: exp close factorized pure}.
\end{proof}

\section{Global disentangling}
\label{sec.short-global-disentangling}

We now pass from the local SRE property to the global one. Fix a pure product state $\phi$.  By Proposition~\ref{prop.cutting}, we can write, for every $x\in\ZZ$, 
\begin{equation}
\label{eq: restate psi locally sre}
    \chi_x:=\psi_{\leq x}\otimes\psi_{>x}
    =\psi\circ\Ad{U^{(x)}},
    \qquad
    \widehat\psi_{>x}:=\phi_{\leq x}\otimes\psi_{>x},  
\end{equation}  
where the cutting unitaries $U^{(x)}$ are $(C,c)$-exponentially anchored at
$x$, uniformly in $x$. 

The proof of the SRE property for $\psi$ proceeds in three steps. First, we remark that the states
$\widehat\psi_{>x}$ satisfy uniformly the assumptions of
Theorem~\ref{thm.connecting-states}. Then, we relate $\widehat\psi_{>x}, \widehat\psi_{>x'}$ for $x\neq x'$ in a convenient way. Finally, we use a powerful Lemma from \cite{kapustin-classification2021} to connect $\psi$ (to be thought of as $\widehat\psi_{>-\infty}$) to the product state $\phi$ (to be thought of as $\widehat\psi_{>\infty}$) via an LGA that is constructed as an infinite product of anchored LGA's. 

\subsection{Uniform estimates for the half-states}

\begin{lemma}\label{lem.short-half-state-estimates}
There are constants $(C,c)$, and $\alpha_0>0$, independent of
$x$, such that every state $\widehat\psi_{>x}$ satisfies:
\begin{enumerate}
    \item $(C,c)$-exponential decay of mutual correlations;
    \item Decay of Schmidt eigenvalues for any finite interval $I$ (as defined in Section \ref{sec.schmidt} and in Proposition \ref{prop.decay-of-schmidt-coeffs}):
    \begin{equation}\label{eq: schmidts decay hat psi}
        \sum_{j\geq\cutoff}\lambda_j^{}
        \leq C\cutoff^{-\alpha_0},
        \qquad \cutoff\geq1.
    \end{equation}
\end{enumerate}
\end{lemma}

\begin{proof}
We first note the following elementary stability property. If a locally SRE
state is transformed by a uniformly exponentially anchored unitary, then the
resulting state is again locally SRE, with modified uniform constants. This is an immediate consequence of Lemma \ref{lem.local-approximation-unitary}.


Since $\psi$ is locally SRE, i.e.\ \eqref{eq: restate psi locally sre} and the anchoring parameters of $U^{(x)}$ are
uniform in $x$, the preceding observation shows that the states $\chi_x$ and $\widehat\psi_x$
are also locally SRE with uniform constants in $x$. 

Furthermore, a $(K,k)$-locally SRE state obviously satisfies $(K,k)$-exponential decay of mutual correlations. Therefore, the first item of Lemma \ref{lem.short-half-state-estimates} is proven. 

Now to the second item. Proposition
\ref{prop:locally-sre-schmidt-tail} therefore gives constants
$C<\infty$ and $\alpha_0>0$, independent of $x$, such that the Schmidt
eigenvalues of $\chi_x$ satisfy the bound \eqref{eq: schmidts decay hat psi}. 
Naturally, this bound passes from $\chi_x$ to $\widehat\psi_{>x}$.

\end{proof}

\subsection{The one-step construction}

We use the following elementary observation for a tripartition $(A,B,C)$ of $\Z$.  If the unit vectors
\[
    \ket{\xi_{AB}}\otimes\ket{\gamma_C},
    \qquad
    \ket{\alpha_A}\otimes\ket{\eta_{BC}}
\]
are within distance $\delta$, then there is a unit vector $\ket{\beta_B}$ such
that both vectors are within distance $3\delta$ of
$\ket{\alpha_A}\otimes\ket{\beta_B}\otimes\ket{\gamma_C}$. To see this, apply the projection
$\kb{\alpha_A}\otimes\mathds1_B\otimes\mathds1_C$ to the first vector.
The result has the form $\ket{\alpha_A}\otimes\ket{b_B}\otimes\ket{\gamma_C}$; normalizing
$\ket{b_B}$ proves the claim by the triangle inequality. We refer to this
observation as the tripartite bridge argument.

\begin{proposition}\label{prop.short-one-step}
There is $\ell_0\in\NN$ such that, for every $\ell\geq\ell_0$ and every
$x\in\ell\ZZ$, there are unitaries
\[
    V^{(x)}\in\AA_{(x,x+\ell]},
    \qquad
    W^{(x)}\in\AA_{>x},
\]
with $W^{(x)}$ $(C,c)$-exponentially anchored at $x+\ell$, such that
\[
    \widehat\psi_{>x}\circ
    \Ad{W^{(x)}V^{(x)}}
    =\widehat\psi_{>x+\ell}.
\]
The constants $C,c$ are independent of $x$ and $\ell$.
\end{proposition}

\begin{proof}
Set $z=x+\ell$ and use the tripartition
\[
    A=(-\infty,x],
    \qquad
    B=(x,z],
    \qquad
    C=(z,\infty).
\]
We work in a GNS representation factorized according to this tripartition,
as in \eqref{eq.split-many-intervals}. 

In that representation,  $\chi_x$ and $\chi_z$ are represented  by unit vectors
\[
    \ket{\chi_x}=\ket{\chi_A}\otimes\ket{\chi_{BC}},
    \qquad
    \ket{\chi_z}=\ket{\chi_{AB}}\otimes\ket{\chi_C}.
\]
After replacing the cutting unitaries by their adjoints if necessary, we may
write
\[
    \ket{\chi_x}=U^{(x)}\ket{\psi},
    \qquad
    \ket{\chi_z}=U^{(z)}\ket{\psi}.
\]

Let $R=\lfloor\ell/4\rfloor$. By
Lemma~\ref{lem.local-approximation-unitary}, there are local unitaries
\[
    u_x\in\AA_{B_R(x)},
    \qquad
    u_z\in\AA_{B_R(z)}
\]
such that
\[
    \norm{U^{(x)}-u_x} \leq Ce^{-c\ell}, \qquad  \norm{U^{(z)}-u_z}  
    \leq Ce^{-c\ell}.
\]
Their supports are disjoint. Hence
\begin{align*}
    \norm{u_x\ket{\chi_z}-u_z\ket{\chi_x}}
    &\leq
      \norm{u_x(\ket{\chi_z}-u_z\ket{\psi})}
      +\norm{u_z(u_x\ket{\psi}-\ket{\chi_x})}\\
    &\leq Ce^{-c\ell}.
\end{align*}
The first vector is product across $AB\mid C$, whereas the second is product
across $A\mid BC$. The tripartite bridge argument therefore gives a unit
vector $\ket{s_B}\in\HH_B$ such that
\begin{equation}\label{eq.short-bridge-vector}
    \norm{u_z\ket{\chi_{BC}}-\ket{s_B}\otimes\ket{\chi_C}}
    \leq Ce^{-c\ell}.
\end{equation}

\begin{lemma}
Let $s_B$ be the associated (to $\ket{s_B}$)  pure state on $\AA_B$.    The states $s_B\otimes\psi_{>z}$ and $\psi_{>x}$ on
$\AA_{>x}$ are $(C,c)$-exponentially close far from $z$, with constants independent
of $x$ and $\ell$.
\end{lemma}
\begin{proof}
First, $u_z\in\AA_{>x}$ and \eqref{eq.short-bridge-vector}
implies
\[
    \norm{(s_B\otimes\psi_{>z})
          -\bigl(\psi_{>x}\circ\Ad{u_z}\bigr)}
    \leq Ce^{-c\ell}.
\]
Moreover, by Lemma \ref{lem.local-approximation-unitary}, $u_z$ remains exponentially anchored at $z$, uniformly in $x$ and $\ell$. 

For $0<r\leq\ell$, we can therefore compare both states on
$\AA_{(x,\infty)\cap B_r(z)^c}$ through
$\psi_{>x}\circ\Ad{u_z}$. The  estimate 
 \eqref{eq.short-bridge-vector} and the anchoring of $u_z$
give
\[
    \norm{\bigl((s_B\otimes\psi_{>z})-\psi_{>x}\bigr)
          \big|_{(x,\infty)\cap B_r(z)^c}}
    \leq C\bigl(e^{-c\ell}+e^{-cr}\bigr)
    \leq Ce^{-cr}.
\]
For $r>\ell$, the region $(x,\infty)\cap B_r(z)^c$ lies to the right of
$z$. There the states are $\psi_{>z}$ and $\psi_{>x}$, respectively.
The definitions of the cut states and the anchoring of $U^{(x)}$ and
$U^{(z)}$ give
\[
    \norm{\psi_{>x}|_{\geq z+r}-\psi_{>z}|_{\geq z+r}}
    \leq Ce^{-cr}.
\]
This concludes the proof of the lemma.
\end{proof}  

We may now apply the half-chain version of
Theorem~\ref{thm.connecting-states} to $\psi_{>x}$ and
$s_B\otimes\psi_{>z}$. Its assumptions for $\psi_{>x}$ follow from
Lemma~\ref{lem.short-half-state-estimates}, while the preceding claim
supplies the required exponential closeness. The half-chain version follows
by applying the same proof to the half-chain algebra. We obtain a unitary
$W^{(x)}\in\AA_{>x}$, with uniform anchoring constants at $z$, such that
\[
    \psi_{>x}\circ\Ad{W^{(x)}}=s_B\otimes\psi_{>z}.
\]

Finally, both $s_B$ and $\phi_B$ are pure states of the finite-dimensional
matrix algebra $\AA_B$. Hence there is a unitary $V^{(x)}\in\AA_B$ such
that
\[
    s_B\circ\Ad{V^{(x)}}=\phi_B.
\]
Since $W^{(x)}$ and $V^{(x)}$ act trivially on $\AA_A$, the last two
displayed equations give the desired identity.
\end{proof}

\subsection{Iteration}

We finish the proof of Theorem~\ref{thm.main}. Choose  $\ell \in N$ sufficiently
large, to be fixed later, and put $x_n=n\ell$.  
 Iterating the one-step
identity of Proposition \ref{prop.short-one-step} (in particular we need already here that $\ell$ is large enough) gives, for every $M\in\NN$,
\begin{equation}\label{eq: recursion1}
    \widehat\psi_{>0}\circ
    \Ad{W^{(x_0)}V^{(x_0)}\cdots
        W^{(x_{M-1})}V^{(x_{M-1})}}
    =\widehat\psi_{>x_M}. 
\end{equation}
Because $W^{(x_m)}\in\AA_{>x_m}$, one has
\begin{align*}
    [V^{(x_n)},W^{(x_m)}]=0
    &\quad\text{whenever }n<m,\\
    [V^{(x_n)},V^{(x_m)}]=0
    &\quad\text{whenever }n\neq m.
\end{align*}
We may therefore move all $V$'s to the right without changing the order of
the $W$'s;
\begin{equation}\label{eq: recursion}
    \widehat\psi_{>0}\circ
    \Ad{W^{(x_0)} \cdots
        W^{(x_{M-1})}}  \Ad{V^{(x_0)}\cdots  V^{(x_{M-1})}}
    =\widehat\psi_{>x_M}. 
\end{equation}
Since $V^{(x_n)}$ are supported on 
disjoint intervals $(x_n,x_{n+1}]$, the expression $\Ad{V^{(x_0)}\cdots  V^{(x_{M-1})}}(A)$, for local $A$, trivially converges to $\alpha_V(A)$ where $\alpha_V$ is an LGA.
Similarly, the expression $  \Ad{W^{(x_0)} \cdots
        W^{(x_{M-1})}} (A)$ converges to $\alpha_W(A)$ where $\alpha_W$ is an LGA. This follows from Lemma~\ref{lem.infiniteproduct} upon choosing $\ell$ large enough, and with the conditions on the factors in the product following from Lemma \ref{lem.anchored-unitary-lga}.
        The convergence is extended by density to $A\in \AA$ and we conclude that 
\[
    \widehat\psi_{>0}\circ\alpha_W\circ\alpha_V=\phi.
\]     
Proceeding in an analogous way, but with the role of left and right reversed, we can relate $\psi_{\leq0}$ by an LGA to
$\phi_{\leq0}$. Finally, we note that 
\[
    \chi_0= \psi_{\leq 0}\otimes \psi_{>0} = \psi\circ\Ad{U^{(0)}},
\]
where the anchored automorphism $\Ad{U^{(0)}}$ is also an LGA, again by Lemma \ref{lem.anchored-unitary-lga}.  Since finite
compositions of LGAs are LGAs, we have now connected $\psi$ to $\phi$ via an LGA, concluding the proof of
   Theorem~\ref{thm.main}.

\section*{Acknowledgments}
W.D.R. and B.O.C. were supported by the FWO and F.R.S.-FNRS under the Excellence of Science (EOS) programme through the research project G0H1122N EOS 40007526 CHEQS, the KULeuven Runner-up Grant No. iBOF DOA/20/011, and the internal KULeuven Grant No. C14/21/086. M.F. was supported by the NSF under
grant DMS-2407290. M.F. thanks A.~Elgart and Y.~Ogata for fruitful discussions. B.O.C.~thanks R.C.~Drumond for fruitful discussions. The authors thank Ayumi Ukai and the referees for pointing out an error in the first version of the paper.

\section*{Data Availability}
Data sharing is not applicable to this article as no new data were created or analysed in this study.

 \section*{Conflict of interest}
 The authors declare no conflict of interest.

\appendix

\section{Uhlmann's Theorem and Fuchs-van de Graaf inequality}

Throughout the appendices, all Hilbert spaces are separable. The \textit{fidelity} between density matrices $\rho, \sigma \in \mathcal B_1(\HH)$ is defined as
\begin{equation}
    F(\rho, \sigma) := \left[\Tr \left(\sqrt{\sqrt \rho \sigma \sqrt \rho}\right)\right]^2.
\end{equation}
It relates to the trace distance (Schatten 1-norm) according to the Fuchs-van de Graaf inequality \cite{fuchs-vandegraaf}:
\begin{equation} \label{eq.fuchs-vandegraaf}
    1 - \sqrt{F(\sigma,\rho)} \le \frac 12 \norm{\rho - \sigma}_1 \le \sqrt{1-F(\sigma,\rho)}. 
\end{equation}
For $\ket \xi \in \HH$, we denote by $\kb{\xi}$ its pure density matrix. A \textit{purification} $\ket \xi$ of $\rho \in \mathcal \mathcal B_1(\HH)$ on $\HH'$ is a unit vector of a Hilbert space $\HH \otimes \HH'$ such that $\Tr_{\HH'} (\kb{\xi}) = \rho$. It is a basic fact that two purifications $\ket \xi, \ket {\xi'}$ of a density matrix $\rho$ on a finite dimensional $\HH'$ are related by a unitary acting on the purifying space $\caH'$. {Consequently, the vectors $\ket \xi, \ket {\xi'}$ carry the same Schmidt rank across the bipartition $\HH \otimes \HH'$}. 

Uhlmann's Theorem \cite{UHLMANN1976273} relates fidelity and purifications:

\begin{theorem}[Uhlmann's Theorem] \label{thm.uhlmann}
    Let $\rho, \sigma \in \mathcal B_1(\HH)$ be finite-rank density matrices on a Hilbert space $\HH$. Let $\ket \xi \in \HH\otimes\HH'$ be a purification of $\rho$ on a finite-dimensional Hilbert space $\HH'$ such that $\dim(\HH') \ge \text{rank}(\sigma)$. Then
    \begin{equation}
        F(\rho, \sigma) = \max \{ |\bra{\xi}\ket{\eta}|^2 \ :\ \Tr_{\HH'}(\ket{\eta} \bra{\eta}) = \sigma\}.
    \end{equation}
 \end{theorem}
\begin{proof}
    Assume, firstly, that $\text{rank}(\rho), \text{rank}(\sigma) < \infty$. 
    Let $P$ be the orthogonal projection in the subspace spanned by $\text{ran}(\rho)$ and $\text{ran} (\sigma)$. By assumption, $P$ is finite. Moreover, it holds that $P\rho P = \rho$ and $(P \otimes \id) \ket{\xi_\rho} = \ket {\xi_\rho}$ for every purification $\ket {\xi_\rho}$ of $\rho$ on $\HH'$. Similarly, $P\sigma P = \sigma$. The density matrices $P\sigma P$ and $P\rho P$ in $\mathcal B_1(P\HH)$, and the purification $(P \otimes \id) \ket \xi \in P\HH \otimes \HH'$ satisfy the assumptions of the finite-dimensional version of Uhlmann's Theorem \cite[Theorem 3.22]{Watrous2018}, which gives us that
    \begin{align*}
        F(P\sigma P, P \rho P) &= \max \{ |\bra{\xi} \ket{\eta}|^2\ :\ \ket \eta \in P\HH \otimes \HH',\  \Tr_{\HH'} (\ket \eta \bra \eta) = P\sigma P \}.
    \end{align*}
    The result then follows from the observation that every purification $\ket{\eta_\sigma}$ of $\sigma$ on $\HH'$ satisfies 
    $$(P\otimes \id)\ket{\eta_\sigma} = \ket{\eta_\sigma}.$$
\end{proof}

A straightforward consequence of Theorem \ref{thm.uhlmann} and inequalities \eqref{eq.fuchs-vandegraaf} is that, under the above conditions, the purification $\ket \eta$ of $\sigma$ can be chosen such that
\begin{align} \label{eq.purifications-closeness}
    \norm{\ket \xi - \ket \eta} \le \norm{\rho-\sigma}_1^{\frac 12}.
\end{align}
A simple application of Uhlmann's Theorem and the Fuchs-van de Graaf inequality allows us to derive results on approximating states by local operations, as stated below. We use the notation $\rho_B = \Tr_{\HH_A} (\rho)$ for the restriction of $\rho$ to $\HH_B$. 

\begin{proposition} \label{prop.krauss}
    Let $\HH_A \otimes \HH_B$ be a bipartite system, with $\dim \HH_A < \infty$, and let $\rho, \sigma \in \mathcal B_1(\HH_A \otimes \HH_B)$ be density matrices such that $\norm{\rho_B - \sigma_B}_1 < \epsilon$, for some $\epsilon>0$. Then, 
    \begin{enumerate} [label = {(\alph*)}]
        \item if there are $\ket \xi, \ket \eta \in \HH_A \otimes \HH_B$ such that $\rho = \kb{\xi}$ and $\sigma = \kb{ \eta}$, there exists a unitary $U \in \caB(\HH_A)$ such that $\norm{U\ket \xi - \ket \eta} < \sqrt\epsilon$. 
        
        \item if $\rho = \kb{\xi}$ is pure, there exists a completely positive trace preserving (CPTP) map $\Phi: \mathcal B_1(\HH_A) \to \mathcal B_1(\HH_A)$ such that
    \begin{equation}
        \norm{(\Phi \otimes \id)(\rho) - \sigma}_1 < 2\sqrt \epsilon. 
    \end{equation}
    This map has a Kraus decomposition
    \begin{equation*}
        \Phi(X) = \sum\limits_{j=1}^m M_j X M_j^*,
    \end{equation*}
    with elements $M_j \in \mathcal B(\HH_A)$ satisfying $\sum\limits_{j=1}^m M_j^* M_j = \mathds 1_{\HH_A}$, and
         $\expval{M_i^* M_j}{\xi} = \delta_{i,j} \nu_j.$
    \end{enumerate}
\end{proposition}
\begin{proof} \hfill 
    \begin{enumerate}  [label = {(\alph*)}]
        \item 
    By the Fuchs-van de Graaf inequality \eqref{eq.fuchs-vandegraaf}, $F(\sigma_B , \rho_B) \ge (1-\frac 12 \epsilon)^2$. Hence, by Theorem \ref{thm.uhlmann}, there exists a purification $\ket {\tilde \eta} \in \HH_A \otimes \HH_B$ of $\sigma_B$ such that $ |\bra{\tilde \eta} \ket{\xi}| \ge 1-\frac 12 \epsilon $. Consequently, 
    $$\norm{\ket {\tilde \eta} - \ket \xi}^2 \le 2 - 2 (1 - \frac \epsilon 2) = \epsilon.$$
    Since both $\ket{\tilde \eta}$ and $\ket \eta$ purify $\sigma_B$, there exists a unitary $U$ on the purifying system $\HH_A$ with the property $U\ket {\tilde \eta} = \ket \eta$. Therefore, $\sqrt\epsilon \ge \norm{U\ket {\tilde \eta} - U\ket \xi} = \norm{\ket \eta - U \ket \xi}$.

    \item 
Let
$
d:=\|\rho_B-\sigma_B\|_1<\varepsilon.
$
Choose a sequence of finite-rank density matrices \(\sigma^{(n)}\) such that
$$
\|\sigma^{(n)}-\sigma\|_1\longrightarrow0.
$$
By contractivity of the partial trace,
$
\sigma_B^{(n)}\longrightarrow\sigma_B
$
in trace norm, and therefore
$
d_n:=\|\rho_B-\sigma_B^{(n)}\|_1\longrightarrow d.
$

For each \(n\), choose a finite-dimensional Hilbert space \(H_{A'_n}\), with
$
\dim H_{A'_n}\geq \operatorname{rank}(\sigma^{(n)}),
$
and a purification
$
|\eta_n\rangle\in H_{A'_n}\otimes H_A\otimes H_B
$
of \(\sigma^{(n)}\). Fix a unit vector \(|e_1\rangle\in H_{A'_n}\). Then
$$
|e_1\rangle\otimes|\xi\rangle
\quad\text{and}\quad
|\eta_n\rangle
$$
are purifications of \(\rho_B\) and \(\sigma_B^{(n)}\), respectively, with finite-dimensional purifying space \(H_{A'_n}\otimes H_A\). By item (a), there exists a unitary
$
U_n\in\mathcal U(H_{A'_n}\otimes H_A)
$
such that
$$
\bigl\|
|\eta_n\rangle-U_n(|e_1\rangle\otimes|\xi\rangle)
\bigr\|
\leq\sqrt{d_n}.
$$

Define the CPTP map
$
\Phi_n(X)
=
\operatorname{Tr}_{A'_n}
\left(
U_n
(|e_1\rangle\langle e_1|\otimes X)
U_n^*
\right).
$
Using contractivity of the trace norm under partial trace and the standard estimate for pure states, we obtain
$$
\|(\Phi_n\otimes 1)(\rho)-\sigma^{(n)}\|_1
\leq 2\sqrt{d_n}.
$$
Since \(H_A\) is finite-dimensional, the set of CPTP maps on \(\mathcal B_1(H_A)\) is compact. Hence, after passing to a subsequence, \(\Phi_n\) converges to a CPTP map \(\Phi\). Taking the limit in the previous estimate gives
$$
\|(\Phi\otimes1)(\rho)-\sigma\|_1
\leq2\sqrt d
<2\sqrt\varepsilon.
$$
Finally, since \(H_A\) is finite-dimensional, \(\Phi\) admits a finite Kraus representation
$$
\Phi(X)=\sum_{j=1}^m M_jXM_j^*.
$$
After a unitary change of Kraus representation, we may assume
$
\langle\xi|M_i^*M_j|\xi\rangle
=
\delta_{ij}\nu_j.
$
Indeed, this follows by diagonalizing the positive Gram matrix
$
G_{ij}=\langle\xi|M_i^*M_j|\xi\rangle.
$
    \end{enumerate}
\end{proof}

A straightforward consequence of the above Proposition is: 

\begin{lemma} \label{lem.large-overlap-with-product}
    Let $\rho = \kb{\xi}$ be a pure density matrix on $\HH_A \otimes \HH_B$, where $\dim \HH_A < \infty$. If 
    $\norm {\rho_B - \sigma}_1 < \epsilon$
    for a pure density matrix $\sigma = \kb{\varphi}$ on $\HH_B$, then there exists a product state $\ket p \in \HH_A \otimes \HH_B$ such that
    \begin{equation}
        | \braket{\xi}{p}| > 1-\epsilon/2.
    \end{equation}
\end{lemma}
\begin{proof}
    We choose an arbitrary pure state $\ket{\varphi'} \in \HH_A$, and apply Item (a) of Proposition \ref{prop.krauss} to $\kb{{\varphi'}} \otimes \kb{{\varphi}}$ and $\kb{{\xi}}$, obtaining a unitary $U \in \mathcal B(\HH_A)$ such that $\norm{(U\ket{\varphi'}) \otimes \ket \varphi - \ket \xi} < \sqrt \epsilon$. The claim holds for $\ket p= (U\ket{\varphi'}) \otimes \ket \varphi$.
\end{proof}

\section{Sopenko's Lemma} \label{appendix.sopenko-lemma}


\noindent Lemma \ref{lemma.tripartite-system} is a corollary of both Uhlmann's theorem and Fuchs-van de Graaf inequality, and it was proven in \cite[Lemma~D.1]{sopenko_chiral_2023} under slightly different assumptions. In what follows, we adapt the proof to our setting. Firstly, an auxiliary lemma:

\begin{lemma} \label{lem.tripartite-system-aux}
    Let $\HH = \HH_A \otimes \HH_B \otimes \HH_C$ be a tripartite system, with $\dim (\HH_A \otimes \HH_B) < \infty$. Consider pure states $\rho = \ket\xi \bra{\xi}$ and $\sigma =\ket\eta \bra{\eta}$ on $\HH$, such that
    \begin{enumerate}
        \item \label{assumption.tripatite-aux-1} $\norm{\ket {\xi} - \ket{\eta}} < \epsilon,$ for some $\epsilon>0$,

        \item \label{assumption.tripatite-aux-2} $\rho_{BC} = \sigma_{BC}$,

        \item \label{assumption.tripatite-aux-3} both states factorize as
        \begin{align*}
            \rho_{AC} = \rho_A \otimes \rho_C, \qquad \sigma_{AC} = \sigma_A \otimes \sigma_C.
        \end{align*}

    \end{enumerate}
    Then there exists a unitary $U \in \UU(\HH_{A} \otimes \HH_B)$ such that $U\ket{\xi} = \ket \eta$ and $\norm{U-\id} \le \epsilon$. 
\end{lemma}

\begin{proof}
    By assumption \ref{assumption.tripatite-aux-2} in the statement of the lemma, there exists a unitary $V_A \in \UU(\HH_A)$ such that $\ket{\eta} = V_A \ket{\xi}$. Hence $\text{rank}(\sigma_A) = \text{rank}(\rho_A)$. Since $\sigma_C = \rho_C$, in particular also $\text{rank}(\sigma_C) = \text{rank}(\rho_C)$. By item \ref{assumption.tripatite-aux-3} and purity of the density matrix $\rho$, it holds true that:
    $$\dim (\HH_B) \ge \text{rank}(\rho_A) \text{rank}(\rho_C).$$ 
    We consider a fixed splitting of $\HH_B$ as 
    \begin{align*}
        \HH_B = \mathcal K \oplus \mathcal K^\perp, 
    \end{align*}
    where $\dim(\mathcal K) = \text{rank}(\rho_A) \text{rank} (\rho_C)$. Moreover, we also consider a splitting of $\mathcal K = \mathcal H_{B_L} \otimes \mathcal H_{B_R}$ with $\dim(\HH_{B_L}) = \text{rank}(\rho_A)$ and $\dim(\HH_{B_R}) = \text{rank}(\rho_C)$. We thus have a further decomposition
    \begin{align*}
        \HH &= (\HH_A \otimes \mathcal K \otimes \HH_C) \oplus (\HH_A \otimes \mathcal K^\perp \otimes \HH_C) \\
        &= (\HH_A \otimes \HH_{B_L} \otimes \HH_{B_R} \otimes \HH_C) \oplus (\HH_A \otimes \mathcal K^\perp \otimes \HH_C) \\
        &= (\HH_L \otimes \HH_R) \oplus (\HH_A \otimes \mathcal K^\perp \otimes \HH_C),
    \end{align*}
    where we have defined $\HH_L = \HH_A \otimes \HH_{B_L}$ and $\HH_R = \HH_{B_R} \otimes \HH_C$. With respect to this decomposition, there exists a purification
    \begin{align*}
        \ket{\xi'} = \ket{\xi}_{L} \otimes \ket{\xi}_{R} \oplus 0
    \end{align*}
    of $\rho_{AC} = \rho_{A} \otimes \rho_C$. Since $\ket{\xi}$ is also a purification of $\rho_{AC}$ on $\HH_B$, there is a unitary $W_B \in \UU(\HH_B)$ such that
    \begin{align*}
        \ket{\xi} = W_B(\ket{\xi}_{L} \otimes \ket{\xi}_{R} \oplus 0).
    \end{align*}
    By using that $[V_A,W_B] = 0$, we can compute:
    \begin{align*}
        \norm{\ket{\xi}_{L} - V_A\ket{\xi}_{L}}_{\HH_L} &= \norm{(\ket{\xi}_{L} \otimes \ket{\xi}_{R} \oplus 0) - (V_A\ket{\xi}_{L}\otimes \ket{\xi}_{R}\oplus 0)}\\
        &= \norm{W_B(\ket{\xi}_{L} \otimes \ket{\xi}_{R}\oplus 0) - V_AW_B(\ket{\xi}_{L} \otimes \ket{\xi}_{R} \oplus 0)} \\
        &= \norm{\ket \xi - \ket{\eta}} \le \epsilon.
    \end{align*}
    Consequently, it is possible to choose $U_L \in \UU(\HH_{L})$ that maps $\ket{\xi}_{L}$ to $V_A\ket{\xi}_{L}$ such that $\norm{U_L-\id} \le \epsilon$. We take 
    $$U = W_B(U_L \otimes \id_{\HH_{B_R}} \oplus \id_{\HH_A\otimes \mathcal K^\perp})W_B^* \in \UU(\HH_A \otimes \HH_B)$$ and finish the proof of the lemma. 
\end{proof}

Below is a generalization of Lemma \ref{lem.tripartite-system-aux} above, without the assumption of exact factorization between regions $AC$ or exact coincidence in regions $BC$. 

\begin{lemma} \label{lemma.tripartite-system}
Let $\HH = \HH_A \otimes \HH_B \otimes \HH_C$ be a tripartite system, with $\dim (\HH_A \otimes \HH_B) < \infty$. Let $\rho = \kb\xi$, $\sigma = \kb\eta$ be pure density matrices on $\HH$, satisfying:
\begin{enumerate}
    \item\label{sopenko-lemma-condition-1} the following proximity bounds:
    \begin{align} \label{eq.proximity-conditions}
        &\norm{\ket \xi - \ket \eta} < {\epsilon}, \qquad \norm{\rho_{BC} - \sigma_{BC}}_1 < \epsilon^5, \qquad \norm{\rho_{AC} - \rho_A \otimes \rho_C}_1 < \frac {\epsilon^{10}} 3,
    \end{align}
    for some $\epsilon \in \left(0,\frac 1{16}\right)$. 
    
    \item\label{sopenko-lemma-condition-2} let $r_{B} := \lfloor\sqrt{\dim(\HH_B)}\rfloor$, and assume
    \begin{align}\label{eq.spectral-concentration-rho-C}
        \sum\limits_{j > r_{B}} \lambda_j^\Gamma \le \frac{\epsilon^{10}}{12},
    \end{align}
    where $\{\lambda_j^\Gamma\}_{j \in \NN}$, are the eigenvalues of $\rho_\Gamma$, $\Gamma=A,C$ (with $\lambda_j^\Gamma \le \lambda_{j-1}^\Gamma$ for all $j$).
\end{enumerate}
Then there exists a unitary $U \in \UU(\HH_{A} \otimes \HH_B)$ such that 
$$\norm{U \ket \xi - \ket \eta} < \frac 12 \epsilon^2$$
and 
$$\norm{U - \mathds 1} < 2\epsilon.$$
\end{lemma}

\begin{proof}
    Consider spectral decompositions
$        \rho_\Gamma = \sum_j \lambda_j^\Gamma \ket{{v_j^\Gamma}}\bra{v_j^\Gamma},$ and let $P_\Gamma = \sum_{j \le r_B} \ket{{v_j^\Gamma}}\bra{v_j^\Gamma}$. Define the states
    \begin{align}
        \hat \rho_\Gamma := \frac{P_\Gamma \rho_\Gamma P_\Gamma}{\Tr(P_\Gamma \rho_\Gamma P_\Gamma)}. 
    \end{align}
    From \eqref{eq.spectral-concentration-rho-C}, it follows that $\norm{\rho_\Gamma - \hat \rho_\Gamma}_1 \le \epsilon^{10}/3$. Thus, from \eqref{eq.proximity-conditions}:
    \begin{align} \label{eq.sopenko-lemma-truncation}
        \norm{\hat \rho_A \otimes \hat \rho_C - \rho_{AC}}_1 \le \epsilon^{10}. 
    \end{align}
    Moreover, $\text{rank}(\hat \rho_\Gamma) \le \sqrt{\dim(\HH_B)}$. We proceed by proving the existence of the unitary $U$.\\ 
    
    Firstly, since $\text{rank}(\hat \rho_A \otimes \hat \rho_C) \le \text{dim}(\HH_B)$, there exists a purification $\ket{\xi'} \in \HH$ of $\hat \rho_A \otimes \hat \rho_C$. Moreover, due to \eqref{eq.sopenko-lemma-truncation} and \eqref{eq.purifications-closeness}, this purification can be chosen such that 
    \begin{align}\label{eq.purification-xi}
        \norm{\ket{\xi'} - \ket \xi} \le \epsilon^5.
    \end{align}
    Let $\rho' = \ket{\xi'}\bra{\xi'}$. It follows that $\norm{\rho_{BC}' - \rho_{BC}}_1 \le 2 \epsilon^5$. Thus,    by  \eqref{eq.proximity-conditions} and a triangle inequality,
    \begin{align*}  \norm{\rho_{BC}' - \sigma_{BC}}_1 \le 3\epsilon^5,
    \end{align*}
  Consequently, by again invoking \eqref{eq.purifications-closeness}, there is a purification $\ket{\eta'} \in \HH$ of $\rho_{BC}'$ such that 
    \begin{align} \label{eq.purification-eta}
        \norm{\ket{\eta'} - \ket\eta} < \sqrt{3\epsilon^5}. 
    \end{align}
    From \eqref{eq.proximity-conditions}, \eqref{eq.purification-xi} and \eqref{eq.purification-eta}, we get that 
    \begin{align}\label{eq.proximity-primed-purifications}
        \norm{\ket {\xi'} - \ket {\eta'}} \le \epsilon +\epsilon^5 + \sqrt{3\epsilon^5} \le 2\epsilon. 
    \end{align}
    We have now two states $\rho'=\ket{\xi'}\bra{\xi'}$ and $\sigma' = \ket{\eta'}\bra{\eta'}$ such that
    \begin{enumerate}
        \item $\norm{\ket{\xi'} - \ket{\eta'} } \le 2\epsilon$;
        \item\label{item.assumption-coincidence} both $\ket{\xi'}$ and $\ket{\eta'}$ are purifications of $\rho'_{BC}$. Hence they are related by a unitary $V_A \in \mathcal B(\HH_A)$. Particularly, $\rho'_{BC} = \sigma'_{BC}$;
        \item there are factorizations
        \begin{align*}
            \rho'_{AC} = \rho'_A \otimes \rho'_C = \hat \rho_A  \otimes \hat \rho_C, \qquad \sigma'_{AC} = \sigma'_A \otimes \sigma'_C. 
        \end{align*}
        The first factorization is by construction of $\ket{\xi'}$. The second follows easily from item \ref{item.assumption-coincidence} above.
    \end{enumerate}
    By Lemma \ref{lem.tripartite-system-aux}, there exists a unitary $U\in \UU(\mathcal H_A \otimes \HH_B)$ such that $U\ket{\xi'} = \ket{\eta'}$ and $\norm{U-\id} < 2\epsilon$. 
    We then obtain
    \begin{align*}
        \norm{U\ket{\xi} - \ket{\eta}} \le \norm{\ket{\xi'} - \ket{\xi}} + \norm{\ket{\eta'} - \ket{\eta}} \le \epsilon^5 + \sqrt{3\epsilon^5} \le \frac 12 \epsilon^2,
    \end{align*}
    finishing the proof.

\end{proof}

\subsection{Proof of Theorem \ref{thm.connecting-states}}

The proof consists of an iterative application of Lemma \ref{lemma.tripartite-system}. 
Without loss, we prove the result for $x=0$. We recall the assumptions on states $\varphi, \varphi'$: there are constants $1<K<\infty$, $k>0$, and $\alpha>0$ such that the following bounds hold:
\begin{enumerate}
    \item $\varphi$ exhibits (K,k)-exponential decay of mutual correlations (as in definition \eqref{eq.def-mutual-exp-decay}). Particularly,
    \begin{equation}\label{eq.proof-sopenko-lemma-mutual-decal}
        \norm{\varphi|_{B_{\ell} \cup B_{2\ell}^c} - \varphi|_{B_\ell} \otimes \varphi|_{B_{2\ell}^c}} \le Ke^{-\frac k2 \ell}.
    \end{equation}
    The above is a consequence of choosing $Y = B_{3\ell/2}$ and applying definition \eqref{eq.def-mutual-exp-decay} with $\ell' = \ell/2$. 

    \item $\varphi$ and $\varphi'$ are (K,k)-exponentially close far from site $0$: 
    \begin{align} \label{eq.proof-sopenko-lemma-exp}
        \norm{\varphi|_{B_\ell^c} - \varphi'|_{B_\ell^c}} \le Ke^{-k\ell}.
    \end{align}

    \item $\varphi$ satisfies a condition on the decay of Schmidt eigenvalues for any finite interval $I$ (as defined in Section \ref{sec.schmidt} and in Proposition \ref{prop.decay-of-schmidt-coeffs}):
\begin{align}\label{eq.decay-schmidt-arbitrary-state-proof}
            \sum\limits_{j\ge \cutoff} \lambda_j \le K(\cutoff)^{-\alpha}, \qquad \cutoff \ge 1.
        \end{align} 
\end{enumerate}

Let $0<p<\frac{1}{16}$ be a positive parameter. We choose
\begin{equation} \label{eq.proof-sopenko-r0}
    r_0 \ge \max\left[1, 1 + \frac{1}{\alpha} \log_2\left( \frac{K}{p^{16}} \right) ,\frac 2k \ln \left( \frac{K}{p^{16}} \right)\right],
\end{equation}
and define the radii $r_n = 2^nr_0$ for $n>0$. We claim that for every $n \in \NN$, there exists a unitary $\hat U_{(n)} \in \AA_{B_{r_{n}}}$ such that  
\begin{align}
    \norm{\varphi_n - \varphi'} < p^{2^{n}}, 
\end{align}
where $\varphi_n = \varphi_{n-1} \circ \Ad{ \hat U_{(n)}}$, and $\varphi_{-1} = \varphi$. Moreover, for $n>0$:
$$\norm{\hat U_{(n)} - \mathds 1} < 2p^{2^{n-1}}.$$
The proof of this claim is by induction on $n$. \\

\noindent \textbf{Proof of the claim for $n=0$.} 
From definition \eqref{eq.proof-sopenko-r0}, we obtain
\begin{align*}
    \norm{\left( \varphi - \varphi' \right)|_{B_{r_0}^c}}  &\le Ke^{-kr_0} \le p^{16}.
\end{align*}
Consider a GNS representation for $\varphi$ that is factorized as $\HH_{B_{r_0}} \otimes \HH_{B_{r_0}^c}$, whose existence is guaranteed by the discussion in Section \ref{sec: split property}. By using the same notation as in Proposition \ref{prop.krauss}, let $\HH_{B_{r_0}} = \HH_A$, and $\HH_{B_{r_0}^c} = \HH_B$. We then have, by Section \ref{sec.norms-states}:
\begin{align*}
    \norm{ \dm{\varphi}|_{B} - \dm{\varphi'}|_{B} }_1 < p^{16}.
\end{align*}
Consequently, by item (a) of Proposition \ref{prop.krauss}, there exists a unitary $U_{(0)} \in \mathcal B(\HH_A)$ such that $\norm{U_{(0)}\ket \varphi - \ket{\varphi'}} < {p^2}/2$. We then have
\begin{align*}
    \norm{\varphi_0 - \varphi'} < p^2,
\end{align*}
where $\varphi_0 = \varphi \circ \Ad{\pi_{\varphi}^{-1}(U_{(0)})}$. Set $\hat U_{(0)} = \pi_{\varphi}^{-1}(U_{(0)})$. Since the GNS representation we chose is factorized, $\hat U_{(0)} \in \AA_{B_{r_0}}$. Notice that inequalities \eqref{eq.proof-sopenko-lemma-mutual-decal} and \eqref{eq.proof-sopenko-lemma-exp} still hold when we exchange $\varphi \to \varphi_0$, for all $\ell \ge r_0$. \\

\noindent \textbf{Induction step.} 
We prove the claim for $n+1$, assuming it for $n \geq 0$. In particular, we assume
\begin{align} \label{eq.sopenko-proof-inductive-step-1}
    \norm{\varphi_n - \varphi'} < p^{2^{n}}, 
\end{align}
where $\varphi_n = \varphi_{n-1} \circ \Ad{ \hat U_{(n)}}$, and $\hat U_{(n)} \in \AA_{B_{r_{n}}}$. Notice that then inequalities \eqref{eq.proof-sopenko-lemma-mutual-decal} and \eqref{eq.proof-sopenko-lemma-exp} still hold when we exchange $\varphi \to \varphi_{n}$, for all $\ell \ge r_{n}$.

Consider a GNS representation of $\varphi_n$ that is factorized as 
$$\HH_{\varphi_n} = \HH_{B_{r_{n}}} \otimes \HH_{B_{r_{n+1}}\setminus B_{r_{n}}} \otimes \HH_{B_{r_{n+1}}^c}.$$
We firstly argue that the pure density matrices of $\ket {\varphi_n}, \ket{\varphi'}$ satisfy the assumptions of Lemma \ref{lemma.tripartite-system}, with
\begin{align*}
    A = B_{r_{n}}, \qquad B = B_{r_{n+1}}\setminus B_{r_{n}}, \qquad C = B_{r_{n+1}}^c,
\end{align*}
and $\epsilon = p^{2^n}$. Equation \eqref{eq.sopenko-proof-inductive-step-1} implies that, with respect to this representation, it holds
\begin{align*}
    \norm{\dm{\varphi_n} - \dm{\varphi'}}_1 < p^{2^{n}}.
\end{align*}
A consequence of the inequality above is that the representatives $\ket{\varphi_n},\ket{\varphi'}$ can be chosen such that 
$$\norm{\ket{\varphi_n} - \ket{\varphi'}} < {p^{2^n}}.$$
Furthermore, given the choice of radii, we obtain from Equation \eqref{eq.proof-sopenko-lemma-exp}:
\begin{align*}
    \norm{(\dm{\varphi_n} - \dm{\varphi'})|_{BC}}_1 &\le Ke^{-2^{n}kr_0} 
     \le (Ke^{-kr_0})^{2^n} \le p^{2^{n+4}} < \epsilon^5,
\end{align*}
where we have used \eqref{eq.proof-sopenko-r0} and the assumption that $K>1$. Moreover, from \eqref{eq.proof-sopenko-lemma-mutual-decal}, we further obtain
\begin{align*}
    \norm{\dm{\varphi_n}|_{AC} - \dm{\varphi_n}|_{A} \otimes \dm{\varphi_n}|_{C} }_1 \le Ke^{-\frac k2 r_{n}} \le \epsilon^{16} < \frac 13\epsilon^{10},
\end{align*}
by choosing $\ell = r_n$ and noticing that $r_{n+1} = 2r_n$.
We now prove that assumption \eqref{eq.spectral-concentration-rho-C} is satisfied. By the assumption that the on-site dimensions are larger than 1 (see comment at the beginning of Section \ref{sec: prelim and notation}), we have the estimate:
\begin{align*}
    r_B = \lfloor \sqrt{\dim{\HH_B}} \rfloor \ge 2^{(r_{n+1} - r_{n})} = 2^{r_n}. 
\end{align*}
By assumption \eqref{eq.decay-schmidt-arbitrary-state}, we have that
\begin{align*}
    \sum\limits_{j > r_{B}} \lambda_j^\Gamma \le \sum\limits_{j > 2^{r_{n}}} \lambda_j^\Gamma \le K(2^{2^{n}r_0})^{-\alpha} \le \frac 1{12} \epsilon^{10},
\end{align*}
{where the last inequality follows from the choice of $r_0$}. Hence we can apply Lemma \ref{lemma.tripartite-system}, obtaining a unitary $U_{(n+1)} \in \UU(\HH_{B_{r_{n}}} \otimes\HH_{B_{r_{n+1}}\setminus B_{r_{n}}})$ such that 
\begin{align} \label{eq.proof-sopenko-step-1-ending}
    \norm{U_{(n+1)} \ket {\varphi_n} - \ket{\varphi'}} \le \frac 12 p^{2^{n+1}} ,
\end{align}
and $\norm{U_{(n+1)} - \mathds 1} < 2p^{2^n}$. Define $\hat U_{(n+1)} = \pi_{\varphi_n}^{-1}(U_{(n+1)}) \in \AA_{B_{r_{n+1}}}$, and $\varphi_{n+1} = \varphi_{n} \circ \Ad{\hat U_{(n+1)}}$. From \eqref{eq.proof-sopenko-step-1-ending}, we get
\begin{align*}
    \norm{\varphi_{n+1} - \varphi'} < p^{2^{n+1}}. 
\end{align*}
This finishes the induction step.

\vspace{.6cm}
 \noindent \textbf{Constructing the unitary $U$.} The norm limit $U = \lim_{n\to \infty} \hat{U}_{(0)} \hat U_{(1)} \dots \hat U_{(n)}$ exists since $\norm{ U_{(n)}-\mathds 1}$ decays fast enough,  and $\varphi' = \lim\limits_{n \to \infty} \varphi_{n} = \varphi \circ \Ad{U}$. For $\ell>0$, let $n_\ell := \max\{ n\ | \ r_{n+1} < \ell \}$. We have that
$$ n_\ell \ge \log_2\left( \frac{\ell}{r_0}\right) - 2.$$
Thus, for $A \in \AA_{B_{\ell}^c}$,
\begin{align*}
    \norm{[A,U]} & \le \lim\limits_{N\to \infty} \sum\limits_{n \ge n_\ell+1}^N 2\norm{A} \norm{\hat U_{(n)}-\mathds 1} \\
    &\le 2\norm{A} \lim\limits_{N\to \infty} \sum\limits_{n \ge n_\ell+1}^N 2p^{2^{n-1}} \le  p^{\frac \ell {4r_0}} \times \text{const.} \times \norm{A}.
\end{align*}
Hence $U \in \AA$ is exponentially anchored at $x=0$.

\section{Control of Schmidt eigenvalues for locally SRE states}

The sole purpose of this appendix is to show that the local SRE property
implies the bound on Schmidt eigenvalues $\lambda_j=\lambda_j^{(I)}$ stated in Proposition
\ref{prop.decay-of-schmidt-coeffs}.    We recall
\[
    d=\max_{x\in\ZZ}\dim(\HH_x)<\infty.
\]

\begin{proposition}\label{prop:locally-sre-schmidt-tail}
Let $\varphi$ be a pure $(K,k)$-locally SRE state.  There are constants
$K'<\infty$ and $\alpha>0$, depending only on $K,k$ and $d$, such that, for
every finite interval $I\subset\ZZ$,
\[
    \sum_{j\geq\cutoff}\lambda_j
    \leq K'\cutoff^{-\alpha},
    \qquad \cutoff\geq1.
\]
In particular, the constants are uniform in the choice of
$I$.
\end{proposition}

\begin{proof}
Let $I=[a,b]$, $A=\{y<a\}$ and $B=\{y>b\}$. Local SRE implies the split
property, so we may represent $\varphi$ by a unit vector
$\ket{\varphi}\in\HH_A\otimes\HH_I\otimes\HH_B$.
At either endpoint, $\ket{\varphi}$ is obtained from a product vector across that cut
by a $(K,k)$-exponentially anchored unitary. Replacing this unitary by its
radius-$R$ approximant from Lemma~\ref{lem.local-approximation-unitary}
gives a vector approximant within distance $2Ke^{-kR}$ of $\ket{\varphi}$, with Schmidt rank at most
$d^{2R}$ across the cut.

We define $P_A$ on $\HH_A$ and $P_B$ on $\HH_B$, as the projectors on the spaces spanned by the Schmidt eigenvectors of these vector approximants. They are hence
each of rank at most $d^{2R}$ and they satisfy
\[
  \norm{(\mathds1-P_A)\ket{\varphi}},\quad
  \norm{(\mathds1-P_B)\ket{\varphi}}\leq2Ke^{-kR}.
\]
Since $P_A,P_B$ commute, we have $\mathds1-P_AP_B\leq  (\mathds1-P_A)+(\mathds1-P_B) $. Therefore, the vector $\zeta_R=(P_A\otimes\mathds1_I\otimes P_B)\ket{\varphi}$ satisfies
\[
  \norm{\ket{\varphi}-\zeta_R}^2
  \leq\norm{(\mathds1-P_A)\ket{\varphi}}^2
       +\norm{(\mathds1-P_B)\ket{\varphi}}^2
  \leq8K^2e^{-2kR}.
\]
The Schmidt rank of $\zeta_R$ across $I\mid I^c$ is at most
$\operatorname{rank}(P_A)\operatorname{rank}(P_B)\leq d^{4R}$.
By the Eckart--Young theorem \cite{EckartYoung1936}, the sum of the Schmidt eigenvalues beyond
rank $d^{4R}$ is the squared distance from $\xi$ to the nearest vector
of Schmidt rank at most $d^{4R}$. Since $\zeta_R$ has this rank, the
preceding estimate gives
\[
  \sum_{j>d^{4R}}\lambda_j\leq8K^2e^{-2kR}.
\]

For $\cutoff\geq2$, choose
$R=\lfloor\log(\cutoff-1)/(4\log d)\rfloor$. This gives
$\sum_{j\geq\cutoff}\lambda_j^{(I)}\leq K'\cutoff^{-\alpha}$
with $\alpha=k/(2\log d)$. Increasing $K'$ covers $\cutoff=1$.
The constants are uniform in $I$.
\end{proof}

\section{Spectral properties of density matrices}

\begin{lemma}\label{lem: two density matrices}
 Let two density matrices $\sigma,
    \omega$ satisfy $\norm{\sigma-\omega}_1 \leq \delta$ and let $|\eta\rangle$ be a normalized eigenvector of $\omega$ with eigenvalue $\lambda >0$.
    Then
     $\norm{P|\eta\rangle-|\eta\rangle} \leq2\delta/\lambda $ with $P$ the spectral projector of $\sigma$ on the interval $[\lambda/2,1]$. 
    \end{lemma}
    \begin{proof}
Set $Q=\mathds1-P$. On $\operatorname{Ran}Q$, the operator
$\lambda-\sigma$ is bounded below by $\lambda/2$. Since $Q$ commutes with
$\sigma$ and $\omega\ket\eta=\lambda\ket\eta$,
\[
  \frac{\lambda}{2}\norm{Q\ket\eta}
  \leq\norm{(\lambda-\sigma)Q\ket\eta}
  =\norm{Q(\omega-\sigma)\ket\eta}
  \leq\norm{\omega-\sigma}\leq\delta.
\]
which proves the claim, since $\norm{\omega-\sigma}\leq \norm{\omega-\sigma}_1$
\end{proof}

\section{Kapustin-Sopenko-Yang's Lemma}

Families of LGAs can be inverted and composed. 
However, a composition of infinitely many LGAs is, in general, not well-defined as an automorphism, let alone an LGA. The following lemma states that an infinite composition is well-defined if the LGAs are anchored in a sufficiently sparse set:

\begin{lemma} \label{lem.infiniteproduct}
        For each $x \in \ZZ$, let $\{\alpha_{G_x,s}\}_{s\in [0,1]}$ be a time-dependent family of automorphisms on $\AA$, such that each generator $G_{x}(s) \in \AA$ is exponentially anchored at site $x \in \ZZ$, and the constants $C<\infty, c>0$ in the definition of the exponential anchoring are uniform with respect to $x \in \ZZ$. Then there is $\ell_0\in \NN$ such that, if $\ell\geq \ell_0$, the limit 
		\begin{equation}\label{eq: product of m lgas}
			\lim_{M\to \infty} \alpha_{G_0,s}  \circ \alpha_{G_\ell,s} \circ \dots \circ\alpha_{G_{M\ell},s}  (A), \qquad A \in \AA
        \end{equation}
      exists and can be written as $\alpha_{G,s}(A)$ where $\{\alpha_{G,s}\}_{s \in [0,1]}$ is the dynamics generated by an exponentially quasi-local time-dependent interaction $\Phi$, as defined in Section \ref{sec.lga}.
\end{lemma}
Lemma \ref{lem.infiniteproduct} is proven in \cite{sopenkoindex} and used (implicitly) in \cite{kapustin-classification2021}. See also \cite[Lemma 4.3]{de_o_carvalho_classification_2025} for a review. 


\section{Anchored unitaries}
\label{app.anchored-unitaries}

\begin{proof}[Proof of Lemma~\ref{lem.local-approximation-unitary}]
For a finite set $Z\subset\ZZ$, let $\mathbb E_Z:\AA\to\AA_Z$ be the
conditional expectation obtained by taking the normalized trace on $Z^c$.
Taking $Z=X_\ell$ and
$A_\ell=\mathbb E_{X_\ell}(U)$, the anchoring assumption gives
\[
  \norm{U-A_\ell}\leq Ke^{-k\ell}=:d_\ell,
\]
using the relation between commutator bounds and partial traces, see e.g.\ \cite{NachtergaeleScholzWerner2013}.

If $d_\ell<1$, the element $A_\ell$ is invertible. Define its polar
unitary by
\[
U_\ell=A_\ell(A_\ell^*A_\ell)^{-1/2}\in\AA_{X_\ell}.
\]
Since $\norm{U-A_\ell}\leq d_\ell$, the spectrum of $|A_\ell|$
lies in $[1-d_\ell,1+d_\ell]$. Consequently,
\[
  \norm{U-U_\ell}
  \leq\norm{U-A_\ell}+\norm{|A_\ell|-\mathds 1}
  \leq2d_\ell.
\]
If $d_\ell\geq1$, set $U_\ell=\mathds 1$. Then the same estimate
follows from $\norm{U-\mathds 1}\leq2$. This proves
\eqref{eq.anchored-unitary-approximants}.

Let $B\in\AA_Y$ and write $r=\dist(X,Y)$. If $r>\ell$, the supports
of $U_\ell$ and $B$ are disjoint, so $[U_\ell,B]=0$. If
$r\leq\ell$ and $d_\ell<1$, then
\begin{align*}
  \norm{[U_\ell,B]}
  &\leq\norm{[U,B]}+2\norm{U_\ell-U}\norm{B}\\
  &\leq\bigl(Ke^{-kr}+4Ke^{-k\ell}\bigr)\norm{B}
   \leq5Ke^{-kr}\norm{B}.
\end{align*}
When $d_\ell\geq1$, the choice $U_\ell=\mathds 1$ has vanishing
commutator. This proves the stated anchoring bound.
\end{proof}

\bibliographystyle{myplainurl.bst}
\bibliography{main}

\end{document}